\documentclass[ejs,preprint,twoside,linksfromyear,nameyear]{imsart}
\usepackage[utf8]{inputenc}
\usepackage{graphicx}
\usepackage{pdfpages}
\usepackage{caption}
\usepackage{subcaption}
\usepackage{amsthm,amsmath,amssymb,bbm,enumerate,eucal,algorithm,algorithmic}
\RequirePackage[authoryear]{natbib}
\RequirePackage[colorlinks,citecolor=blue,urlcolor=blue]{hyperref}
\usepackage{breakurl}

\doi{10.1214/13-EJS771}
\pubyear{2013}
\volume{7}
\issue{0}
\firstpage{264}
\lastpage{291}

\startlocaldefs

\long\def\@makecaption#1#2{%
  \vskip\abovecaptionskip
  \footnotesize
  \sbox\@tempboxa{\itshape\textsc{#1}. #2}%
  \ifdim \wd\@tempboxa >\hsize
    \itshape\textsc{#1}. #2\par
  \else
    \global \@minipagefalse
    \hb@xt@\hsize{\hfil\box\@tempboxa\hfil}%
  \fi
  \vskip\belowcaptionskip}

\setattribute{figurecaption}{size}{\footnotesize\itshape}
\setattribute{figurename}   {size}{\scshape}
\setattribute{figurename}   {skip}{.~}

\numberwithin{equation}{section}
\newtheorem{theo}{Theorem}[section]
\newtheorem{lemma}{Lemma}[section]
\newtheorem{model}{Model}

\def\ie{\emph{i.e.}}

\DeclareMathOperator{\E}{\mathbb{E}}
\DeclareMathOperator{\lse}{\mathrm{LSE}}

\DeclareMathOperator{\defin}{\overset{\mathrm{def}}{=}}
\DeclareMathOperator{\1}{\mathbbm{1}}
\newcommand{\vvec}[1]{\mathbf{#1}}
\DeclareMathOperator{\R}{\mathbb{R}}
\DeclareMathOperator{\N}{\mathbb{N}}
\DeclareMathOperator{\K}{\mathcal{KL}}
\DeclareMathOperator{\B}{\mathcal{B}}
\DeclareMathOperator{\F}{\mathbb{F}}
\DeclareMathOperator{\Ft}{\mathcal{F}}
\DeclareMathOperator{\V}{\mathcal{V}}
\DeclareMathOperator{\M}{\mathcal{M}}
\DeclareMathOperator{\D}{\mathcal{D}}
\DeclareMathOperator{\m}{\mathbf{m}}
\DeclareMathOperator{\MP}{\mathcal{M}_{+,\pi}^1}
\DeclareMathOperator{\X}{\mathbf{X}}
\DeclareMathOperator{\Y}{\mathbf{Y}}
\DeclareMathOperator{\x}{\mathbf{x}}
\DeclareMathOperator{\Proba}{\mathbb{P}}
\DeclareMathOperator{\ProbaData}{\mathcal{P}}

\DeclareMathOperator{\agg}{\textsc{a}}

\newcounter{hyp}
\newcounter{cond}
\newenvironment{hyp}{\refstepcounter{hyp}\begin{itemize}
  \item[({\bf A\arabic{hyp}})]}{\end{itemize}}

\endlocaldefs

\begin{document}

\begin{frontmatter}

% "Title of the paper"
\title{PAC-Bayesian estimation and prediction in sparse additive models}
\runtitle{PAC-Bayesian sparse additive prediction}

% indicate corresponding author with \corref{}
% \author{\fnms{John} \snm{Smith}\corref{}\ead[label=e1]{smith@foo.com}\thanksref{t1}}
% \thankstext{t1}{Thanks to somebody}
% \address{line 1\\ line 2\\ printead{e1}}
% \affiliation{Some University}
%
\begin{aug}
\author{\fnms{Benjamin} \snm{Guedj}\corref{}\thanksref{t1}\ead[label=e1]{benjamin.guedj@upmc.fr}}
\thankstext{t1}{Corresponding author.}
\address{Laboratoire de Statistique Th\'eorique et Appliqu\'ee\\
  Universit\'e Pierre et Marie Curie - UPMC\\
  Tour 25 - 2\`eme \'etage, bo\^ite n$^\circ$ 158\\
  4, place Jussieu\\
  75252 Paris Cedex 05, France\\
  \printead{e1}}
\end{aug}
\medskip
\textbf{\and}
\begin{aug}
\author{\fnms{Pierre}  \snm{Alquier}\thanksref{t2}\ead[label=e2]{pierre.alquier@ucd.ie}}
\thankstext{t2}{Research partially supported by the French ``Agence Nationale pour la Recherche''
under grant ANR-09-BLAN-0128 ``PARCIMONIE''.}
 \address{School of Mathematical Sciences\\
 University College Dublin\\
 Room 528 - James Joyce Library\\
 Belfield, Dublin 4, Ireland \\
 \printead{e2}}

    \runauthor{B. Guedj and P. Alquier}
\end{aug}

\begin{abstract}
  The present paper is about estimation and prediction in high-dimensional
  additive models under a sparsity assumption ($p\gg n$ paradigm). A PAC-Bayesian strategy
  is investigated, delivering oracle inequalities in probability. The
  implementation is performed through recent outcomes in
  high-dimensional MCMC algorithms, and the performance of our method
  is assessed on simulated data.
\end{abstract}

\received{\smonth{8} \syear{2012}}

\tableofcontents

\end{frontmatter}
\vfill\eject

\section{Introduction}

Substantial progress has been achieved over the last years in
estimating high-dimensional regression models. A thorough introduction
to this dynamic field of contemporary statistics is provided by the
recent monographs \citet{B:htf,B:bv2g}. In the popular
framework of linear and generalized linear models, the Lasso estimator introduced
by \citet{A:tibshirani} immediately proved successful. Its theoretical properties
have been extensively studied and
its popularity has never wavered since then, see for example \citet{B:btw,A:v2g,A:brt,A:my}. However, even though
numerous phenomena are well captured within this linear context,
restraining high-dimensional statistics to this setting is
unsatisfactory. To relax the strong assumptions required in the linear
framework, one idea is to investigate a more general class of models,
such as nonparametric regression models of the
form $Y=f(X)+W$, where $Y$ denotes the response, $X$ the
predictor and $W$ a zero-mean noise. A good compromise between complexity and effectiveness
is the additive model. It has been extensively studied and formalized
for thirty years now. Amongst many other references, the reader is
invited to refer to \citet{A:stone,A:ht,B:ht,B:hardle}. The core of
this model is that the regression
function is written as a sum of univariate functions $f=\sum_{i=1}^p
f_i$, easing its interpretation. Indeed, each covariate's effect is
assessed by a unique function. This class of
nonparametric models is a popular setting in statistics, despite the
fact that
classical estimation procedures are known to perform poorly as soon as
the number of covariates $p$ exceeds the number of observations
$n$ in that setting.

In the present paper, our goal is to investigate a PAC-Bayesian-based
prediction strategy in the high-dimensional additive framework ($p\gg n$
paradigm). In that context, estimation is essentially possible at the price of a sparsity assumption, \ie, most of the $f_i$ functions are zero. More precisely, our
setting is non-asymptotic. As
empirical evidences of sparse representations accumulate,
high-dimensional statistics are more and more coupled with a sparsity
assumption, namely that the intrinsic dimension $p_0$ of the data is
much smaller than $p$ and $n$, see e.g. \citet{A:ghv}. Additive modelling under a sparsity
constraint has been essentially studied under the scope of the Lasso in \citet{A:mv2gb}, \citet{A:ss2012} and \citet{A:ky2010} or of a combination of
functional grouped Lasso and backfitting algorithm in
\citet{A:rllw}. Those papers inaugurated the
study of this problem and contain essential theoretical results consisting in
asymptotics (see \citet{A:mv2gb,A:rllw}) and non-asymptotics (see \citet{A:ss2012,A:ky2010}) oracle inequalities. The present article should be seen as a
constructive contribution towards a deeper understanding of prediction
problems in the additive framework. It should also be stressed that our
work is to be seen as an attempt to relax as much as possible
assumptions made on the model, such as restrictive conditions on the
regressors' matrix. We consider them too much of a
non-realistic burden
when it comes to prediction problems.\looseness=1

Our \emph{modus operandi} will be based on PAC-Bayesian
results, which is original in that context to our knowledge. The
PAC-Bayesian theory originates in the two seminal papers
\citet{A:stw,A:mcallester} and has been extensively formalized in the
context of classification (see \citet{B:catoni2004,B:catoni2007}) and
regression (see
\citet{A:audibert2004,PhD:audibert,PhD:alquier,A:alquier,A:ac2010,A:ac2011}). However,
the methods presented in these references are not explicitly designed to cover
the high-dimensional setting under the sparsity assumption. Thus, the
PAC-Bayesian theory has
been worked out in the sparsity perspective lately, by \citet{A:dt2008,A:dt2012a,A:al,A:rt}. The main message of these studies is that aggregation with
a properly chosen prior is able to deal effectively with the sparsity
issue. Interesting additional references addressing the aggregation outcomes
would be
\citet{PhD:rigollet,A:audibert2009}. The former aggregation procedures rely on an exponential weights approach,
achieving good statistical properties. Our method should be seen as an
extension of these techniques, and is particularly focused on additive
modelling specificities. Contrary to procedures such as the Lasso, the Dantzig
selector and other penalized methods which are provably consistent
under restrictive assumptions on the Gram matrix associated to the predictors,
PAC-Bayesian aggregation requires only minimal assumptions on the
model. Our method is supported by oracle inequalities in
probability, that are valid in both asymptotic and non-asymptotic
settings. We also show that our estimators achieve the optimal rate
of convergence over traditional smoothing classes such as Sobolev ellipsoids.
It should be stressed that our work is inspired by \citet{A:ab}, which addresses the celebrated
single-index model with similar
tools and philosophy. Let us also mention that although the use of PAC-Bayesian techniques are
original in this context, parallel work has been conducted in the deterministic
design case by \citet{A:suzuki}.\looseness=-1

A major difficulty when considering high-dimensional problems
is to achieve a favorable compromise between statistical and
computational performances. The recent and thorough monograph
\citet{B:bv2g} shall provide the reader with valuable insights that address this
drawback. As a consequence, the explicit implementation of PAC-Bayesian techniques
remains unsatisfactory as existing routines are only put to test with small values
of $p$ (typically $p<100$), contradicting with the high-dimensional
framework. In the meantime, as a solution of a convex problem the
Lasso proves computable for large values of $p$ in reasonable amounts of time. We therefore
focused on improving the computational aspect of our PAC-Bayesian
strategy. Monte Carlo Markov Chains (MCMC) techniques proved increasingly
popular in the Bayesian community, for they probably are the best way
of sampling from potentially complex probability
distributions. The reader willing to find a thorough introduction to
such techniques is invited to refer to the comprehensive monographs
\citet{B:mr,B:mt}. While \citet{A:ab,A:al} explore versions of the
reversible jump MCMC method (RJMCMC) introduced by \citet{A:green},
\citet{A:dt2008,A:dt2012a} investigate a Langevin-Monte
Carlo-based method, however only a deterministic design is
considered. We shall try to overcome those limitations by considering
adaptations of
a recent procedure whose comprehensive description is to be found in
\citet{PhD:petralias,A:pd}. This procedure called Subspace Carlin and
Chib algorithm originates in the seminal paper by \citet{A:cc}, and has a
close philosophy of \citet{A:hdw}, as it favors local moves for the
Markov chain. We provide numerical evidence that our method is
computationally efficient, on simulated data.

The paper is organized as follows. \autoref{S:maths} presents our PAC-Bayesian prediction strategy in additive models. In particular,
it contains the main theoretical results of this paper which consist in oracle inequalities. \autoref{S:mcmc}
is devoted to the implementation of our procedure, along with
numerical experiments on simulated data, presented in \autoref{S:simus}. Finally,
and for the sake of clarity, proofs have been postponed to
\autoref{S:proof}.

\section{PAC-Bayesian prediction}\label{S:maths}

Let $(\Omega,\mathcal{A},\Proba)$ be a probability space on which we denote by $\{(\X_i,Y_i)\}_{i=1}^n$ a sample of $n$ independent
and identically distributed (i.i.d.) random vectors in $(-1,1)^p\times\R$, with $\X_i = (X_{i1},\dots,X_{ip})$, satisfying
\begin{equation*}\label{eq:model}
Y_i=\psi^\star(\X_i)+\xi_i=\sum_{j=1}^p\psi_j^\star(X_{ij})+\xi_i, \quad i\in\{1,\dots,p\},
\end{equation*}
where $\psi_1^\star,\dots,\psi_p^\star$ are $p$ continuous functions $(-1,1)\to \R$ and $\{\xi_i\}_{i=1}^n$ is a set of i.i.d. (conditionaly to $\{(\X_i,Y_i)\}_{i=1}^n$) real-valued random variables. Let $\ProbaData$ denote the distribution of the sample $\{(\X_i,Y_i)\}_{i=1}^n$.
Denote by $\E$ the expectation computed with respect to $\Proba$ and let $\|\cdot\|_\infty$ be the supremum norm. We make the two following assumptions.
\begin{hyp}\label{As:hyp1}
  For any integer $k$, $\E[|\xi_1|^k]<\infty$, $\E[\xi_1|\X_1]=0$ and there exist two positive
  constants $L$ and $\sigma^2$ such that for any integer $k\geq 2$,
  $$
  \E[|\xi_1|^k|\X_1]
  \leq \frac{k!}{2}\sigma^2L^{k-2}.
  $$
\end{hyp}
\begin{hyp}\label{As:hyp2}
  There exists a constant
  $C>\max(1,\sigma)$ such that $\|\psi^\star\|_\infty \leq C$.
\end{hyp}
Note that \autoref{As:hyp1} implies that $\E \xi_1=0$ and
that the distribution of
$\xi_1$ may depend on $\X_1$.
In particular, \autoref{As:hyp1} holds if
$\xi_1$ is a zero-mean gaussian with variance $\gamma^2(\X_1)$ where $x\mapsto\gamma^2(x)$ is bounded.

Further, note that the boundedness assumption \autoref{As:hyp2} plays a
key role in our approach, as it allows to use a version of
Bernstein's inequality which is one of the two main technical
tools we use to state our results. This assumption is not only a
technical prerequisite since it proved crucial for critical
regimes: indeed, if the intrinsic dimension $p_0$ of the
regression function $\psi^\star$ is still large, the boundedness
of the function class allows much faster estimation rates. This
point is profusely discussed in \citet{A:rwy}.

We are mostly interested in sparse additive models, in which only a few $\{\psi^\star_j\}_{j=1}^p$ are not identically zero.
Let
$\{\varphi_k\}_{k=1}^\infty$ be a known countable set of
continuous functions $\R\to (-1,1)$
called the dictionary. In the sequel, $|\mathcal{H}|$ stands for the cardinality of a set $\mathcal{H}$.
For any $p$-th tuple $\m=(m_1,\dots,m_p)\in\N^p$, denote by $S(\m)\subset\{1,\dots,p\}$ the set of indices of nonzero elements of $\m$, \ie,
\begin{equation*}
|S(\m)|=\sum_{j=1}^p\1[m_j>0],
\end{equation*}
and define
\begin{equation*}
 \Theta_{\m}=\left\{\theta\in\R^{m_1}\times\dots\times\R^{m_p}\right\},
\end{equation*}
with the convention $\R^0=\emptyset$. The set $\Theta_{\m}$ is embedded with its canonical Borel field $\B(\Theta_{m})=\B(\R^{m_1})\otimes\dots\otimes\B(\R^{m_p})$. Denote by
\begin{equation*}
\Theta\defin\bigcup_{\m\in\M}\Theta_{\m},
\end{equation*}
which is equipped with the $\sigma$-algebra
$\mathcal{T}=\sigma\left(\bigvee_{\m\in\M}\B(\Theta_{\m})\right)$,
where $\M$ is the collection of models $\M=\{\m=(m_1,\dots,m_p)\in\N^p\}$.
Consider the span of the set $\{\varphi_k\}_{k=1}^\infty$, \ie, the set of functions
\begin{equation*}
 \F=\left\{\psi_\theta=\sum_{j\in S(\m)} \psi_j=\sum_{j\in S(\m)}\sum_{k=1}^{m_j}\theta_{jk}\varphi_k \colon \theta\in\Theta_{\m}, \m\in\M\right\},
\end{equation*}
equipped with a countable generated $\sigma$-algebra denoted by $\Ft$. The risk and empirical risk associated to any $\psi_\theta\in\F$ are
defined respectively as
\begin{equation*}
  R(\psi_\theta) = \E[Y_1-\psi_\theta(\X_1)]^2 \quad \text{ and } \quad R_n(\psi_\theta)=r_n(\{\X_i,Y_i\}_{i=1}^n,\psi_\theta),
\end{equation*}
where
\begin{equation*}
  r_n(\{\x_i,y_i\}_{i=1}^n,\psi_\theta)= \frac{1}{n}\sum_{i=1}^n\left(y_i-\psi_\theta(\x_i)\right)^2.
\end{equation*}
Consider the probability $\eta_{\alpha}$ on the set $\M$ defined by
\begin{equation*}
\eta_{\alpha}\colon \m\mapsto \frac{1-\frac{\alpha}{1-\alpha}}{1-\left(\frac{\alpha}{1-\alpha}\right)^{p+1}}\binom{p}{|S(\m)|}^{-1}\alpha^{\sum_{j=1}^p m_j},
\end{equation*}
for some $\alpha\in(0,1/2)$. Let us stress the fact that the probability $\eta_{\alpha}$ acts as a penalization term over a model $\m$, on the number of its active regressors through the combinatorial term $\binom{p}{|S(\m)|}^{-1}$ and on their expansion through $\alpha^{\sum_{j=1}^p m_j}$.

Our procedure relies on the following construction of the probability $\pi$, referred to as the prior, in order to promote the sparsity properties of the target regression function
$\psi^\star$. For any $\m\in\M$, $\zeta>0$ and $\x\in\Theta_{\m}$, denote by $\B^1_{\m}(\x,\zeta)$ the $\ell^1$-ball centered in $\x$ with radius $\zeta$. For any $\m\in\M$, denote by $\pi_{\m}$ the uniform distribution on $\B^1_{\m}(0,C)$. Define the probability $\pi$ on $(\Theta,\mathcal{T})$,
\begin{equation*}
\pi(A)=\sum_{\m\in\M}\eta_{\alpha}(\m)\pi_{\m}(A), \quad A\in\mathcal{T}.
\end{equation*}
Note that the volume $V_{\m}(C)$ of $\B^1_{\m}(0,C)$ is given by
\begin{equation*}
V_{\m}(C)=\frac{(2C)^{\sum_{j\in S(\m)} m_j}}{\Gamma\left({\sum_{j\in S(\m)} m_j}+1\right)}
  = \frac{(2C)^{\sum_{j\in S(\m)} m_j}}{\left({\sum_{j\in S(\m)} m_j}\right)!}.
\end{equation*}
Finally, set $\delta > 0$ (which may be interpreted as an inverse temperature parameter) and the posterior Gibbs transition density is
\begin{multline}\label{eq:def-gibbs-theta}
\rho_\delta(\{(\x_i,y_i)\}_{i=1}^n,\theta) \\ =\sum_{\m\in\M}\frac{\eta_{\alpha}(\m)}{V_{\m}(C)}\1_{\B_{\m}^1(0,C)}(\theta)\frac{\exp[-\delta r_n(\{\x_i,y_i\}_{i=1}^n,\psi_\theta)]}{\int \exp[-\delta r_n(\{\x_i,y_i\}_{i=1}^n,\psi_\theta)]\pi(\mathrm{d}\theta)}.
\end{multline}
We then consider two competing estimators. The first one is the
randomized Gibbs estimator $\hat{\Psi}$, constructed with parameters $\hat{\theta}$ sampled from the posterior Gibbs density, \ie, for any $A\in\Ft$,
\begin{equation}\label{eq:randomized}
  \Proba(\hat{\Psi}\in A | \{\X_i,Y_i\}_{i=1}^n)=\int_A \rho_\delta(\{\X_i,Y_i\}_{i=1}^n,\theta)\pi(\mathrm{d}\theta),
\end{equation}
while the second one is the aggregated Gibbs estimator
$\hat{\Psi}^{\agg}$ defined as the posterior mean
\begin{equation}\label{eq:aggregated}
  \hat{\Psi}^{\agg} = \int \psi_\theta \rho_\delta(\{\X_i,Y_i\}_{i=1}^n,\theta)\pi(\mathrm{d}\theta) = \E[\hat{\Psi}|\{\X_i,Y_i\}_{i=1}^n].
\end{equation}
These estimators have been introduced in
\citet{B:catoni2004,B:catoni2007} and investigated in further work by \citet{A:audibert2004,PhD:alquier,A:alquier,A:dt2008,A:dt2012a}.

For the sake of clarity, denote by $\D$ a generic numerical constant in the sequel.
We are now in a position to write a PAC-Bayesian oracle inequality.
\begin{theo}\label{T:additivemodels}
  Let $\hat{\psi}$ and
  $\hat{\psi}^{\agg}$ be realizations of the Gibbs
  estimators defined by \eqref{eq:randomized}--\,\eqref{eq:aggregated}, respectively. Let \autoref{As:hyp1} and \autoref{As:hyp2} hold. Set $w=8C\max(L,C)$ and
$\delta = n\ell/[w+4(\sigma^2+C^2)]$,
for $\ell\in(0,1)$, and let $\varepsilon\in(0,1)$. Then with $\Proba$-probability at least $1-2\varepsilon$,
  \begin{multline}\label{eq:additivemodels}
\left. \begin{array}{l}
R(\hat{\psi})- R(\psi^\star)
\\ R(\hat{\psi}^{\agg})- R(\psi^\star)
\end{array} \right\}
        \leq \D \underset{\m\in\M}{\inf}\ \underset{\theta\in\B_{\m}^1(0,C)}{\inf}
        \left\{ R(\psi_{\theta}) - R(\psi^\star) \vphantom{\frac{1}{2}} \right. \\ \left. +|S(\m)|\frac{\log(p/|S(\m)|)}{n}+\frac{\log(n)}{n}\sum_{j\in S(\m)}m_j+\frac{\log(1/\varepsilon)}{n} \right\},
\end{multline}
where $\D$ depends upon $w$, $\sigma$, $C$, $\ell$ and $\alpha$ defined above.
\end{theo}
Under mild assumptions, \autoref{T:additivemodels} provides inequalities which admit the following interpretation.
If there exists a ``small'' model in the collection $\M$, \ie, a model $\m$ such that $\sum_{j\in S(\m)}m_j$ and $|S(\m)|$ are small, such that $\psi_{\theta}$ (with $\theta\in\Theta_{\m}$) is close to $\psi^\star$, then $\hat{\psi}$ and $\hat{\psi}^{\agg}$ are also close to $\psi^\star$ up to $\log(n)/n$ and $\log(p)/n$ terms. However, if no such model exists, at least one of the terms $\sum_{j\in S(\m)}m_j/n$ and $|S(\m)|/n$ starts to emerge, thereby deteriorating the global quality of the bound. A satisfying estimation of $\psi^\star$ is typically possible when $\psi^\star$ admits a sparse representation.

%As explained in \citet{A:ab}, it is still possible to derive oracle inequalities for our PAC-Bayesian estimators \eqref{eq:randomized}-\eqref{eq:aggregated}, even if \ref{As:hyp1} no longer holds. Indeed, \citet{A:s_al} provides a
%concentration inequality which is an alternative to Bernstein's
%inequality, holding with less restrictive assumptions. However,
%using this inequality would deteriorates the quality of the
%right-hand term in inequalities
%\eqref{eq:additivemodels}--\eqref{T:additivemodels:aggregated}.
%
%\citet{A:ac2011} investigates the properties of PAC-Bayesian
%estimates under an assumption which is weaker than \autoref{As:hyp1}. Nevertheless, their results hold only for the
%linear models, and to our knowledge no extension exists for
%our setting so far.

To go further, we derive from \autoref{T:additivemodels} an inequality on
Sobolev ellipsoids. We show that our
procedure achieves the optimal rate of convergence in this setting.
For the sake of shortness, we consider Sobolev spaces, however one can easily derive the following results in other functional spaces such as Besov spaces. See \citet{B:tsybakov} and the references therein.

The notation $\{\varphi_k\}_{k=1}^\infty$ now refers to the
(non-normalized) trigonometric system, defined as
\begin{equation*}
  \varphi_1\colon t \mapsto 1, \quad \varphi_{2j}\colon t \mapsto
  \cos(\pi j t), \quad \varphi_{2j+1}\colon t \mapsto
  \sin(\pi j t),
\end{equation*}
with $j\in\N^*$ and $t\in(-1,1)$.
Let us denote by
$S^\star$ the set of indices of non-identically zero regressors. That
is, the regression function $\psi^\star$ is
\begin{equation*}
  \psi^\star=\sum_{j\in S^\star}\psi_j^\star.
\end{equation*}
Assume that for any $j\in S^\star$, $\psi_j^\star$ belongs to the Sobolev ellipsoid
$\mathcal{W}(r_j,d_j)$ defined as
\begin{equation*}
  \mathcal{W}(r_j,d_j) = \left\{f\in \mathrm{L}^2([-1,1])\colon
    f=\sum_{k=1}^\infty\theta_{k}\varphi_k \quad\mathrm{and}\quad
    \sum_{i=1}^\infty i^{2r_j}\theta_{i}^2\leq d_j \right\}.
\end{equation*}
with $d_j$ chosen such that $\sum_{j\in S^\star}\sqrt{d_j}\leq C\sqrt{6}/\pi$ and
for unknown regularity parameters
$r_1,\dots,r_{|S^\star|}\geq 1$.
Let us stress the
fact that this assumption casts our results onto the adaptive setting. It also implies that $\psi^\star$ belongs to the
Sobolev ellipsoid $\mathcal{W}(r,d)$, with
$r=\min_{\substack{j\in S^\star}}r_j$ and
$d=\sum_{j\in S^\star}d_j$
% \max_{\substack{j\in S^\star}}d_j$
, \ie,
\begin{equation}\label{eq:sob}
  \psi^\star=\sum_{j\in
    S^\star}\sum_{k=1}^\infty \theta_{jk}^\star\varphi_k.
\end{equation}
It is worth pointing out that in that setting, the Sobolev ellipsoid
is better approximated by the $\ell^1$-ball
$\B_{\m}^1(0,C)$ as the dimension of $\m$ grows.
Further, make the following assumption.
\begin{hyp}\label{As:hyp3}
  The distribution of the data $\ProbaData$ has a probability density with respect to the
  corresponding Lebesgue measure, bounded from above by a
  constant $B>0$.
\end{hyp}
\begin{theo}\label{T:sob}
  Let $\hat{\psi}$ and
  $\hat{\psi}^{\agg}$ be realizations of the Gibbs
  estimators defined by \eqref{eq:randomized}--\,\eqref{eq:aggregated}, respectively. Let \autoref{As:hyp1}, \autoref{As:hyp2} and \autoref{As:hyp3} hold. Set $w=8C\max(L,C)$ and
$\delta = n\ell/[w+4(\sigma^2+C^2)]$,
for $\ell\in(0,1)$, and let $\varepsilon\in(0,1)$. Then with $\Proba$-probability at least
  $1-2\varepsilon$,
  \begin{multline*}
  \left. \begin{array}{l}
R(\hat{\psi})- R(\psi^\star)
\\ R(\hat{\psi}^{\agg})- R(\psi^\star)
\end{array} \right\}
    \leq \\ \D
    \left\{ \sum_{j\in S^\star} d_j^{\frac{1}{2r_j+1}}\left(\frac{\log(n)}{2nr_j}\right)^{\frac{2r_j}{2r_j+1}} + |S^\star|\log(p/|S^\star|)/n
     + \frac{\log(1/\varepsilon)}{n} \right\},
  \end{multline*}
  where $\D$ is a constant depending only on $w$, $\sigma$, $C$, $\ell$, $\alpha$ and $B$.
\end{theo}
\autoref{T:sob} illustrates that we obtain the minimax rate of
convergence over Sobolev classes up to a $\log(n)$ term. Indeed, the
minimax rate to estimate a single function with regularity $r$ is
$n^{\frac{2r}{2r+1}}$, see for example \citet[Chapter
2]{B:tsybakov}. \autoref{T:additivemodels} and \autoref{T:sob} thus
validate our method.

A salient fact about \autoref{T:sob} is its links with existing work: assume
that all the $\psi_j^\star$ belong to the same Sobolev ellipsoid
$\mathcal{W}(r,d)$. The convergence rate is now
$\log(n)n^{-\frac{2r}{2r+1}}+\log(p)/n$.
This rate (down to a
$\log(n)$ term) is the same as the one exhibited by \citet{A:ky2010} in the context of  multiple kernel learning ($n^{-\frac{2r}{2r+1}}+\log(p)/n$). \citet{A:ss2012} even obtain faster rates which
correspond to smaller functional spaces. However, the results presented
by both \citet{A:ky2010} and \citet{A:ss2012} are obtained under
stringent conditions on the design, which are not necessary to prove \autoref{T:sob}.

A natural extension is to consider sparsity on both regressors and their expansion, instead of sparse regressors and nested expansion as before. That is, we no longer consider the first $m_j$ dictionary functions for the expansion of regressor $j$. To this aim, we slightly extend the previous notation. Let $K\in\N^*$ be the length of the dictionary. A model is now denoted by $\m=(\m_1,\dots,\m_p)$ and for any $j\in\{1,\dots,p\}$, $\m_j=(m_{j1},\dots,m_{jK})$ is a $K$-sized vector whose entries are $1$ whenever the corresponding dictionary function is present in the model and $0$ otherwise. Introduce the notation
\begin{equation*}
S(\m)=\{\m_j\neq\vvec{0}, j\in\{1,\dots,p\}\},\quad S(\m_j)=\{m_{jk}\neq 0, k\in\{1,\dots,K\}\}.
\end{equation*}
The prior distribution on the models space $\M$ is now
\begin{equation*}
\eta_\alpha\colon\m\mapsto \frac{1-\alpha\frac{1-\alpha^{K+1}}{1-\alpha}}{1-\left(\alpha\frac{1-\alpha^{K+1}}{1-\alpha}\right)^{p+1}}\binom{p}{|S(\m)|}^{-1}\prod_{j\in S(\m)}\binom{K}{|S(\m_j)|}^{-1}\alpha^{|S(\m_j)|},
\end{equation*}
for any $\alpha\in(0,1/2)$.
\begin{theo}\label{T:coro}
Let $\hat{\psi}$ and
  $\hat{\psi}^{\agg}$ be realizations of the Gibbs
  estimators defined by \eqref{eq:randomized}--\,\eqref{eq:aggregated}, respectively. Let \autoref{As:hyp1} and \autoref{As:hyp2} hold. Set $w=8C\max(L,C)$ and
$\delta = n\ell/[w+4(\sigma^2+C^2)]$,
for $\ell\in(0,1)$, and let $\varepsilon\in(0,1)$.\vadjust{\eject} Then with $\Proba$-probability at least $1-2\varepsilon$,
\begin{multline*}
\left. \begin{array}{l}
R(\hat{\psi})- R(\psi^\star)
\\ R(\hat{\psi}^{\agg})- R(\psi^\star)
\end{array} \right\}
        \leq \D \underset{\m\in\M}{\inf}\ \underset{\theta\in\B_{\m}^1(0,C)}{\inf}
        \left\{ R(\psi_{\theta}) - R(\psi^\star) \vphantom{\frac{1}{2}} \right. \\ \left. +|S(\m)|\frac{\log(p/|S(\m)|)}{n}+\frac{\log(nK)}{n}\sum_{j\in S(\m)}|S(\m_j)|+\frac{\log(1/\varepsilon)}{n} \right\},
\end{multline*}
where $\D$ depends upon $w$, $\sigma$, $C$, $\ell$ and $\alpha$ defined above.
\end{theo}

\section{MCMC implementation}\label{S:mcmc}

In this section, we describe an implementation of the method outlined in the previous section. Our goal is to sample from the Gibbs
posterior distribution $\rho_{\delta}$. We use a version of the so-called Subspace Carlin and Chib (SCC) developed by \citet{PhD:petralias,A:pd} which originates in the Shotgun
Stochastic Search algorithm (see \citet{A:hdw}). The key idea of the algorithm lies in a stochastic search heuristic that restricts moves in neighborhoods of the visited models. Let $T\in\N^*$ and denote by $\{\theta(t),\m(t)\}_{t=0}^T$ the Markov chain of interest, with $\theta(t)\in\Theta_{\m(t)}$. Define $i\colon t\mapsto \{+,-,=\}$, the three possible moves performed by the algorithm: an addition, a deletion or an adjustment of a regressor. Let $\{\vvec{e}_1,\dots,\vvec{e}_p\}$ be the canonical base of $\R^p$. For any model $\m(t)=(m_1(t),\dots,m_p(t))\in\M$, define its neighborhood $\{\V^+[\m(t)],\V^-[\m(t)],\V^=[\m(t)]\}$, where
\begin{equation*}
\V^+[\m(t)]=\{\vvec{k}\in\M\colon \vvec{k}=\m(t)+x\vvec{e}_j, x\in\N^*, j\in\{1,\dots,p\}\backslash S[\m(t)]\},
\end{equation*}
\begin{equation*}
\V^-[\m(t)]=\{\vvec{k}\in\M\colon \vvec{k}=\m(t)-m_j(t)\vvec{e}_j, j\in S[\m(t)]\},
\end{equation*}
and
\begin{equation*}
\V^=[\m(t)]=\{\vvec{k}\in\M\colon S(\vvec{k})=S[\m(t)]\}.
\end{equation*}
A move $i(t)$ is chosen with probability $q[i(t)]$. By convention, if $S[\m(t)]=p$ (respectively $S[\m(t)]=1$) the probability of performing an addition move (respectively a deletion move) is zero. Note $\xi\colon\{+,-\}\mapsto\{-,+\}$ and let $D_{\m}$ be the design matrix in model $\m\in\M$. Denote by $\lse_{\m}$ the least square estimate $\lse_{\m}=(D_{\m}^\prime D_{\m})^{-1}D_{\m}^\prime\Y$ (with $\Y=(Y_1,\dots,Y_n)$) in model $\m\in\M$. For ease of notation, let $\mathcal{I}$ denote the identity matrix. Finally, denote by $\phi(\cdot;\mu,\Gamma)$ the density of a Gaussian distribution $\mathcal{N}(\mu,\Gamma)$ with mean $\mu$ and covariance matrix $\Gamma$. A description of the full algorithm is presented in \autoref{alg}.
\begin{algorithm}[t]
\caption{A Subspace Carlin and Chib-based algorithm}
\label{alg}
\begin{algorithmic}[1]
\STATE Initialize $(\theta(0),\m(0))$.
\FOR{$t=1$ to $T$}
%\STATE Generate $\theta_{\m(t-1)}$ from the truncated Gaussian distribution
%$\mathcal{N}(\lse_{\m_{t-1}},\delta^{-1}(\tra{D_{\m_{t-1}}}D_{\m_{t-1}})^{-1})$ on the ball $\B^1_{\m_{t-1}}(0,C)$.
 \STATE Choose a move $i(t)$ with probability $q[i(t)]$.
 \STATE For any $\vvec{k}\in\V^{i(t)}[\m(t-1)]$, generate $\theta_{\vvec{k}}$ from the proposal density $\phi(\cdot;\lse_{\vvec{k}},\sigma^2 \mathcal{I})$.
 \STATE Propose a model $\vvec{k}\in\V^{i(t)}[\m(t-1)]$ with probability
 \begin{equation*}
 \gamma(\m(t-1),\vvec{k})=\frac{A_{\vvec{k}}}{\sum_{\vvec{j}\in\V^{i(t)}[\m(t-1)]}A_{\vvec{j}}},
 \end{equation*}
 where
% \begin{equation*}
% A_{\vvec{j}}=\rho_\delta(\theta_{\vvec{j}})\eta_{\alpha}(\vvec{j})\1_{\B_{\vvec{j}}^1(0,C)}(\theta_{\vvec{j}})\prod_{\vvec{h}\neq\vvec{j}}\phi(\theta_{\vvec{h}};\lse_{\vvec{h}},\sigma^2 \mathcal{I}).
% \end{equation*}
 \begin{equation*}
 A_{\vvec{j}}=\frac{\rho_\delta(\theta_{\vvec{j}})}{\phi(\theta_{\vvec{j}};\lse_{\vvec{j}},\sigma^2 \mathcal{I})}.
 \end{equation*}
 \IF{$i(t)\in\{+,-\}$}
 \STATE For any $\vvec{h}\in\V^{\xi(i(t))}[\vvec{k}]$, generate $\theta_{\vvec{h}}$ from the proposal density $\phi(\cdot;\lse_{\vvec{h}},\sigma^2 \mathcal{I})$. Note that $\m(t-1)\in\V^{\xi(i(t))}[\vvec{k}]$.
 \STATE Accept model $\vvec{k}$, \ie, set $\m(t)=\vvec{k}$ and $\theta(t)=\theta_{\vvec{k}}$, with probability
 \begin{multline*}
 \alpha=\min\left(1,\frac{A_{\vvec{k}}q[i(t)]\gamma(\vvec{k},\m(t-1))}{A_{\m(t-1)}q[\xi(i(t))]\gamma(\m(t-1),\vvec{k})}\right) \\ =\min\left(1, \frac{q[i(t)]\sum_{\vvec{h}\in \V^{i(t)}[\m(t-1)]} A_{\vvec{h}}}{q[\xi(i(t))]\sum_{\vvec{h}\in\V^{\xi(i(t))}[\vvec{k}]} A_{\vvec{h}}} \right).
 \end{multline*}
 Otherwise, set $\m(t)=\m(t-1)$ and $\theta(t)=\theta_{\m(t-1)}$.
 \ELSE
 \STATE Generate $\theta_{\m(t-1)}$ from the proposal density $\phi(\cdot;\lse_{\m(t-1)},\sigma^2 \mathcal{I})$.
 \STATE Accept model $\vvec{k}$, \ie, set $\m(t)=\vvec{k}$ and $\theta(t)=\theta_{\vvec{k}}$, with probability
 \begin{equation*}
 \alpha=\min\left(1,\frac{A_{\vvec{k}}\gamma(\vvec{k},\m(t-1))}{A_{\m(t-1)}\gamma(\m(t-1),\vvec{k})}\right).
 \end{equation*}
 Otherwise, set $\m(t)=\m(t-1)$ and $\theta(t)=\theta_{\m(t-1)}$.
 \ENDIF
\ENDFOR
\end{algorithmic}
\end{algorithm}

The estimates $\hat{\Psi}$ and
$\hat{\Psi}^{\agg}$ are obtained as
\begin{equation*}
  \hat{\Psi} = \sum_{j=1}^p\sum_{k=1}^K\theta_{jk}(T)\varphi_k,
\end{equation*}
and for some burnin $b\in\{1,\dots,T-1\}$,
\begin{equation*}
  \hat{\Psi}^{\agg} =
  \sum_{j=1}^p\sum_{k=1}^K\left(\frac{1}{T-b}\sum_{\ell = b+1}^T\theta_{jk}(\ell)\right)\varphi_k.
\end{equation*}
The transition kernel of the chain defined above is reversible with respect to $\rho_\delta\otimes\eta_{\alpha}$, hence this procedure ensures that $\{\theta(t)\}_{t=1}^T$ is a Markov
Chain with stationary distribution $\rho_\delta$.

\section{Numerical studies}\label{S:simus}

In this section we validate the effectiveness of our method
on simulated data.
All our numerical studies have
been performed with the software R (see \citet{M:R}).
The method is available on the CRAN website (\url{http://www.cran.r-project.org/web/packages/pacbpred/index.html}), under the
name \texttt{pacbpred} (see \citet{R:pacbpred}).

Some comments are in order here about how to calibrate the constants $C$, $\sigma^2$, $\delta$ and $\alpha$. Clearly, a too small value for $C$ will stuck the algorithm, preventing
the chain to escape from the initial model. Indeed, most proposed models will be discarded since the acceptance ratio will frequently take the value $0$. Conversely, a
large value for $C$ deteriorates the quality of the bound
in \autoref{T:additivemodels}, \autoref{T:sob}, \autoref{T:coro} and \autoref{T:regression}. However, this only influences the
theoretical bound, as its contribution to the acceptance ratio is limited to $\log(2C)$. We
thereby proceeded with typically large values of $C$ (such
as $C=10^{6}$). As the parameter $\sigma^2$ is the variance of the proposal distribution $\phi$, the practioner should tune it in accordance with the noise level of the data. The parameter requiring the finest calibration is $\delta$: the convergence of the algorithm is sensitive to its choice. \citet{A:dt2008,A:dt2012a} exhibit the theoretical value $\delta = n/4\sigma^2$. This value leads to very good numerical performances, as it has been also noticed by \citet{A:dt2008,A:dt2012a,A:ab}. The choice for $\alpha$ is guided by a similar reasoning to the one for $C$. Its contribution to the acceptance ratio is limited to a $\log(1/\alpha)$ term. The value $\alpha=0.25$ was used in the simulations for its apparent good properties. Although it would be computationally costly, a finer calibration through methods such as cross-validation is possible.

Finally and as a general rule, we strongly encourage
practitioners to run several chains of inequal lengths and
to adjust the number of iterations needed by observing if
%$R_n(\Psi_{\hat{\theta}_t})$
the empirical risk is stabilized.

\begin{table}[t]
\caption{Each number is the mean (standard deviation) of the RSS over 10 independent runs}
\label{table:rss}
\begin{tabular}{c|ccc}
 & $p=50$ & $p=200$ & $p=400$  \\
 MCMC & $3000$ it. & $10000$ it. & $20000$ it. \\
\hline
 \autoref{mod1} & 0.0318 (0.0047) & 0.0320 (0.0029) & 0.0335 (0.0056) \\
 \autoref{mod2} & 0.0411 (0.0061) & 0.1746 (0.0639) & 0.2201 (0.0992) \\
 \autoref{mod3} & 0.0665 (0.0421) & 0.1151 (0.0399) & 0.1597 (0.0579) %\\
% \autoref{mod4} & &
\end{tabular}
\end{table}
\begin{model}\label{mod1}
$n=200$ and $S^\star=\{1,2,3,4\}$. This model is similar to \citet[Section 3, Example 1]{A:mv2gb} and is given by
\begin{equation*}
Y_i=\psi^\star_1(X_{i1})+\psi^\star_2(X_{i2})+\psi^\star_3(X_{i3})+\psi^\star_4(X_{i4})+\xi_i,
\end{equation*}
with
\begin{multline*}
\psi^\star_1\colon x\mapsto -\sin(2x), \quad \psi^\star_2\colon x\mapsto x^3, \quad \psi^\star_3\colon x\mapsto x, \\ \psi^\star_4\colon x\mapsto e^{-x}-e/2, \quad \xi_i\sim\mathcal{N}(0,0.1),\quad i\in\{1,\dots,n\}.
\end{multline*}
\end{model}
The covariates are sampled from independent uniform distributions over $(-1,1)$.
\begin{model}\label{mod2}
$n=200$ and $S^\star=\{1,2,3,4\}$. As above but correlated. The covariates are sampled from a multivariate gaussian distribution with covariance matrix $\Sigma_{ij}=2^{-|i-j|-2}$, $i$, $j\in\{1,\dots,p\}$.
\end{model}
\begin{model}\label{mod3}
$n=200$ and $S^\star=\{1,2,3,4\}$. This model is similar to \citet[Section 3, Example 3]{A:mv2gb} and is given by
\begin{equation*}
Y_i=5\psi^\star_1(X_{i1})+3\psi^\star_2(X_{i2})+4\psi^\star_3(X_{i3})+6\psi^\star_4(X_{i4})+\xi_i,
\end{equation*}
with
\begin{multline*}
\psi^\star_1\colon x\mapsto x, \quad \psi^\star_2\colon x\mapsto 4(x^2-x-1), \quad \psi^\star_3\colon x\mapsto \frac{\sin(2\pi x)}{2-\sin(2\pi x)}, \\ \psi^\star_4\colon x\mapsto 0.1\sin(2\pi x)+0.2\cos(2\pi x) + 0.3\sin^2(2\pi x)+0.4\cos^3(2\pi x) \\ +0.5\sin^3(2\pi x), \quad \xi_i\sim\mathcal{N}(0,0.5),\quad i\in\{1,\dots,n\}.
\end{multline*}
The covariates are sampled from independent uniform distributions over $(-1,1)$.
\end{model}
%\begin{model}\label{mod4}
%$n=200$ and $S^\star=\{1,\dots,12\}$. This model is in the spirit of \citet[Section 3, Example 4]{A:mv2gb} and is given by the same functions as in \autoref{mod3}.
%\begin{multline*}
%Y_i=\psi^\star_1(X_{i1})+\psi^\star_2(X_{i2})+\psi^\star_3(X_{i3})+\psi^\star_4(X_{i4}) \\ + 1.5\left[\psi^\star_1(X_{i5})+\psi^\star_2(X_{i6})+\psi^\star_3(X_{i7})+\psi^\star_4(X_{i8})\right] \\ +2\left[\psi^\star_1(X_{i9})+\psi^\star_2(X_{i10})+\psi^\star_3(X_{i11})+\psi^\star_4(X_{i12})\right] +\xi_i,
%\end{multline*}
%with $\xi_i\sim\mathcal{N}(0,0.5184)$, $i\in\{1,\dots,n\}$. The covariates are sampled from independent uniform distributions over $(-1,1)$.
%\end{model}

\begin{figure}[b!]
\begin{subfigure}[t]{.45\textwidth}
\caption{\autoref{mod1}, $p=200$.}
\includegraphics[width = \linewidth]{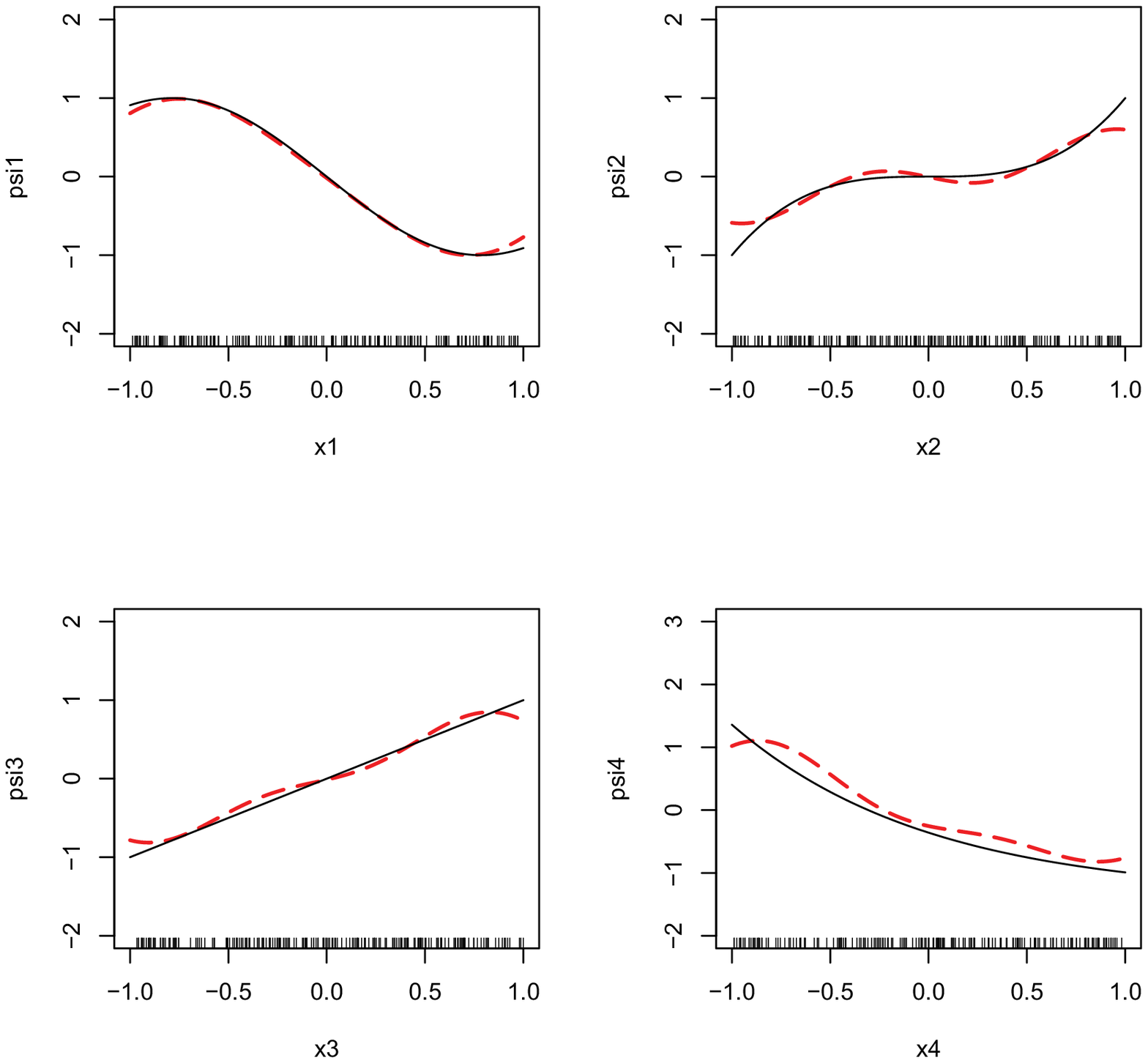}
\end{subfigure}
\begin{subfigure}[t]{.45\textwidth}
\caption{\autoref{mod1}, $p=400$.}
\includegraphics[width = \linewidth]{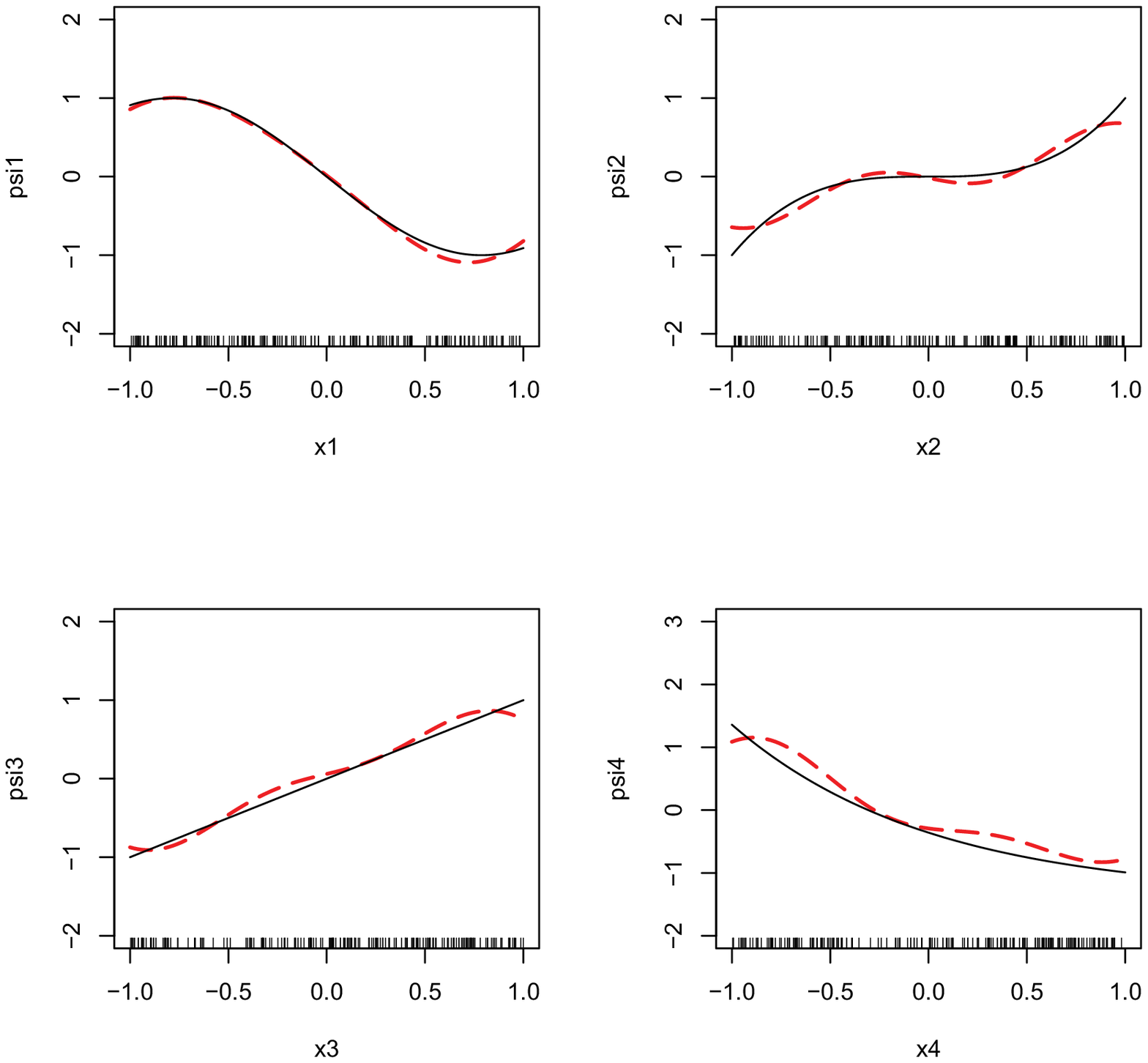}
\end{subfigure}
\begin{subfigure}[t]{.45\textwidth}
\caption{\autoref{mod2}, $p=50$.}
\includegraphics[width = \linewidth]{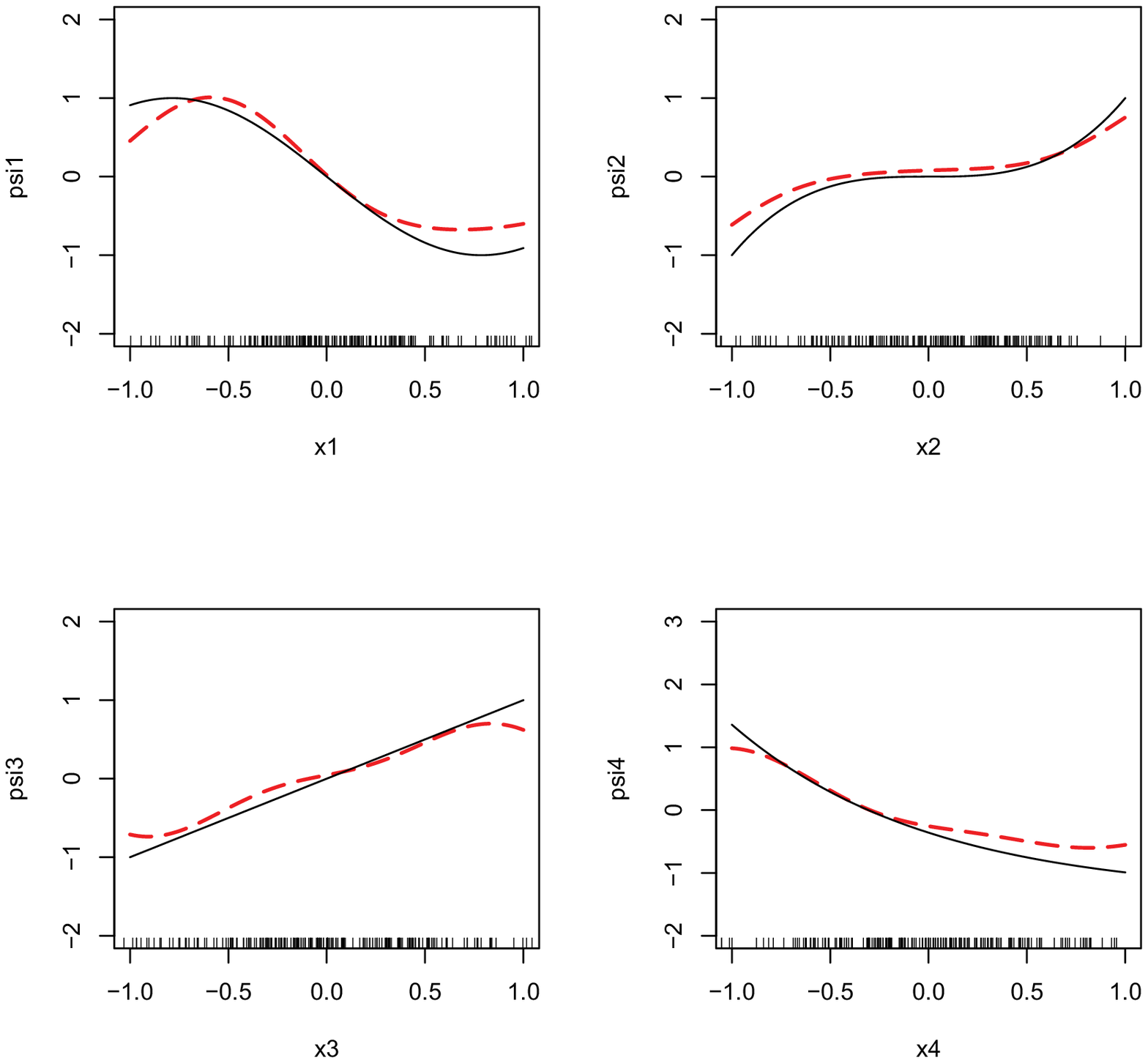}
\end{subfigure}
\begin{subfigure}[t]{.45\textwidth}
\caption{\autoref{mod3}, $p=50$.}
\includegraphics[width = \linewidth]{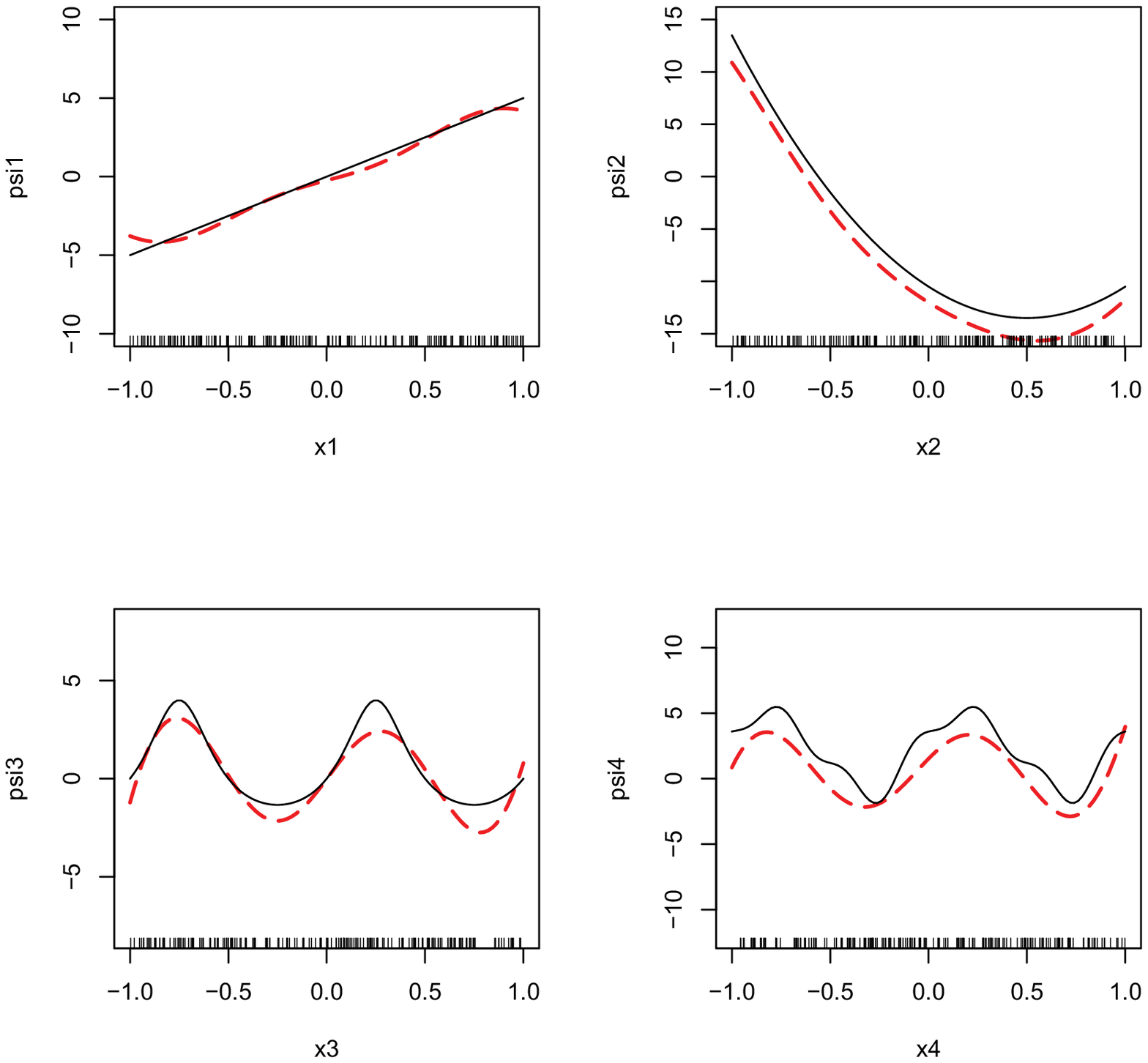}
\end{subfigure}
\caption{Estimates (red dashed lines) for $\psi^\star_1$, $\psi^\star_2$, $\psi^\star_3$ and $\psi^\star_4$ (solid black lines). Other estimates (for $\psi^\star_j$, $j\notin\{1,2,3,4\}$) are mostly zero.}
\label{fig1}
\vspace*{-12pt}
\end{figure}

\begin{figure}[t]
\begin{subfigure}[t]{.45\textwidth}
\caption{\autoref{mod1}, $p=200$.}
\center
\includegraphics[width = .6\linewidth]{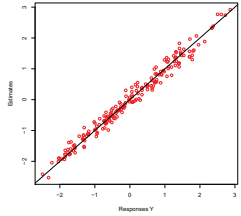}
\end{subfigure}
\begin{subfigure}[t]{.45\textwidth}
\caption{\autoref{mod1}, $p=400$.}
\center
\includegraphics[width = .6\linewidth]{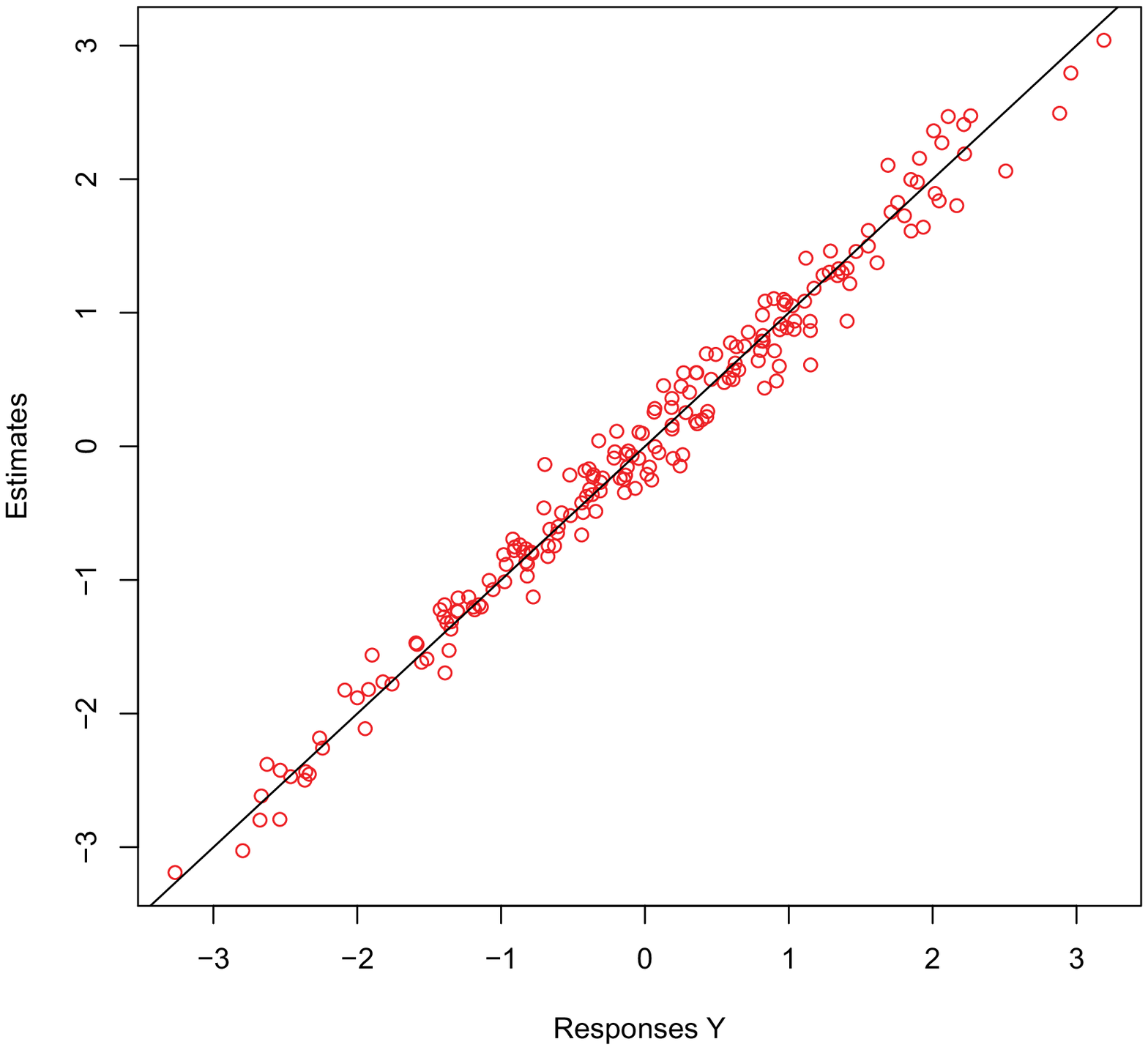}
\end{subfigure}
\begin{subfigure}[t]{.45\textwidth}
\caption{\autoref{mod2}, $p=50$.}
\center
\includegraphics[width = .6\linewidth]{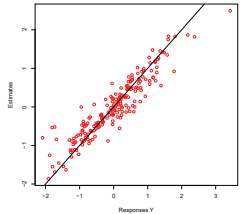}
\end{subfigure}
\begin{subfigure}[t]{.45\textwidth}
\caption{\autoref{mod3}, $p=50$.}
\center
\includegraphics[width = .6\linewidth]{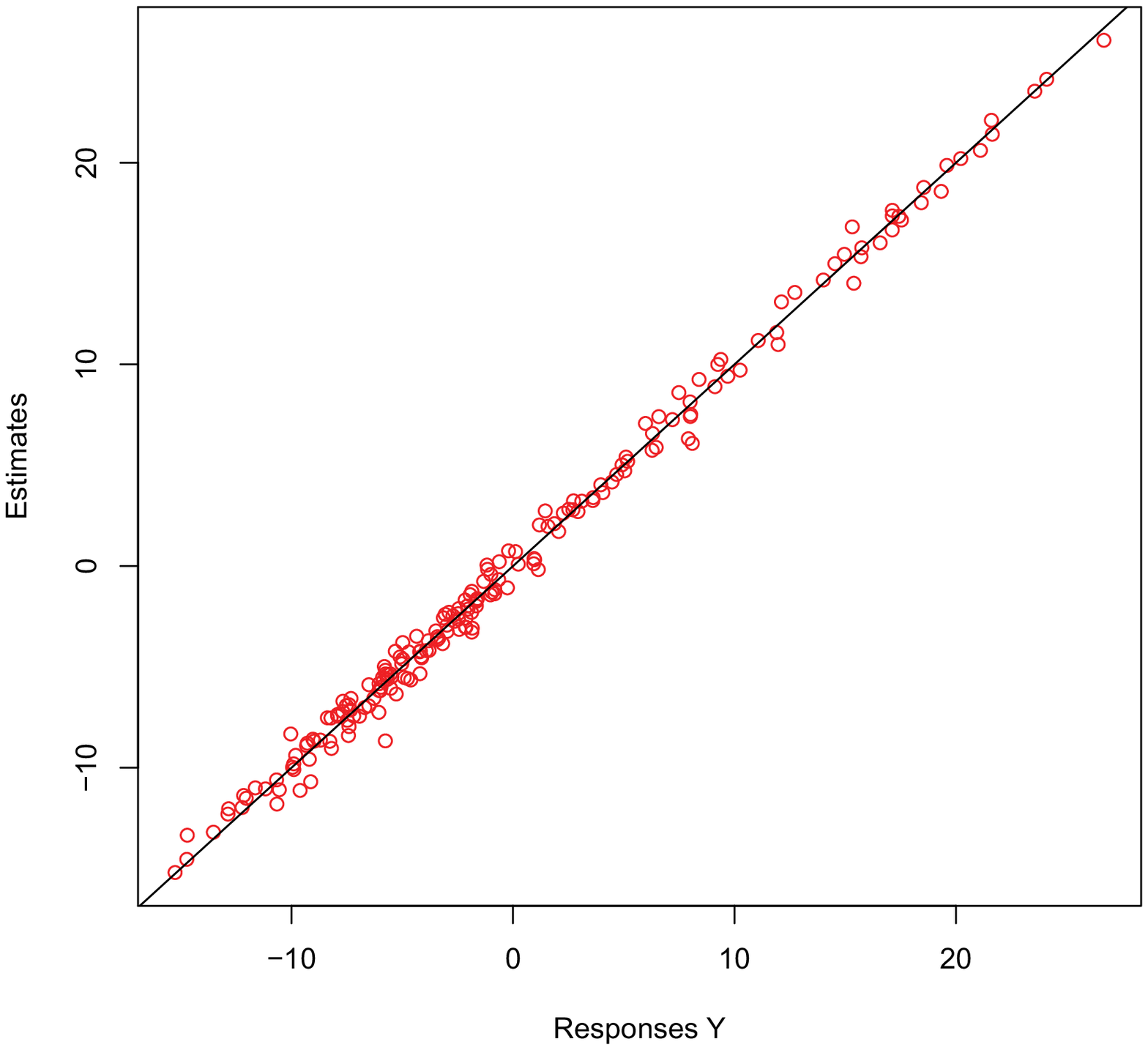}
\end{subfigure}
\caption{plot of the responses $Y_1,\dots,Y_n$ against their estimates. The more points on the first bisectrix (solid black line), the better the estimation.}
\label{fig2}
\end{figure}

The results of the simulations are summarized in \autoref{table:rss} and illustrated by \autoref{fig1} and \autoref{fig2}. The reconstruction of the true regression function $\psi^\star$ is achieved even in very high-dimensional situations, pulling up our method at the level of the gold standard Lasso.

\section{Proofs}\label{S:proof}

To start the chain of proofs leading to \autoref{T:additivemodels},
\autoref{T:sob} and \autoref{T:coro}, we recall and prove some lemmas to establish \autoref{T:regression} which consists in a general PAC-Bayesian inequality in the spirit of \citet[Theorem 5.5.1]{B:catoni2004} for classification or \citet[Lemma 5.8.2]{B:catoni2004} for regression. Note also that \citet[Theorem 1]{A:dt2012a} provides a similar inequality in the deterministic design case. A salient fact on \autoref{T:regression} is that the validity of the oracle inequalities only involves the distribution of the noise variable $\xi_1$, and that distribution is independent of the sample size $n$.

The proofs of the following two classical results are omitted. \autoref{L:massart} is a
version of Bernstein's inequality which originates in \citet[Proposition 2.19]{B:massart}, whereas
\autoref{L:catoni} appears in
\citet[Equation 5.2.1]{B:catoni2004}.

For $x\in\R$,
denote $(x)_+=\max(x,0)$. Let $\mu_1$, $\mu_2$ be two probabilities. The Kullback-Leibler divergence of $\mu_1$ with
respect to $\mu_2$ is denoted $\K(\mu_1,\mu_2)$ and is
\begin{equation*}
  \K(\mu_1,\mu_2) =
  \begin{cases}
    \int \log\left(\frac{\mathrm{d}\mu_1}{\mathrm{d}\mu_2}\right)\mathrm{d}\mu_1 &
    \mathrm{if\ } \mu_1 \ll \mu_2, \\
    \infty & \mathrm{otherwise.}
  \end{cases}
\end{equation*}
Finally, for any measurable space $(A,\mathcal{A})$ and any probability $\pi$ on $(A,\mathcal{A})$, denote by $\MP(A,\mathcal{A})$ the set of probabilities on $(A,\mathcal{A})$ absolutely continuous with respect to $\pi$.
\begin{lemma}\label{L:massart}
  Let $(T_i)_{i=1}^n$ be independent real-valued
  variables. Assume that there exist two positive constants
  $v$ and $w$ such that, for any integer $k \geq 2$,
  \begin{equation*}
    \sum_{i=1}^n\E[(T_i)_+^k] \leq \frac{k!}{2}vw^{k-2}.
  \end{equation*}
  Then for any $\gamma \in \left(0,\frac{1}{w}\right)$,
  \begin{equation*}
    \E\left[\exp\left(\gamma\sum_{i=1}^n(T_i-\E T_i)\right)\right] \leq \exp\left(\frac{v\gamma^2}{2(1-w\gamma)}\right).
  \end{equation*}
\end{lemma}
\begin{lemma}\label{L:catoni}
  Let $(A,\mathcal{A})$ be a measurable space. For any probability $\mu$ on $(A,\mathcal{A})$ and any measurable function
  $h : A \to \R$ such that $\int(\exp\circ\, h) \rm{d}\mu < \infty$,
  \begin{equation*}
    \log\int(\exp\circ\, h) \mathrm{d}\mu = \underset{m\in\MP(A,\mathcal{A})}{\sup}
    \int h \mathrm{d}m - \K(m,\mu),
  \end{equation*}
  with the convention
  $\infty-\infty =-\infty$. Moreover, as soon as $h$ is upper-bounded on the
  support of $\mu$, the supremum with respect to $m$ on the right-hand
  side is reached for the Gibbs distribution $g$ given by
  \begin{equation*}
    \frac{\mathrm{d}g}{\mathrm{d}\mu}(a) =
    \frac{\exp(h(a))}{\int(\exp\circ\, h)\mathrm{d}\mu}, \quad a\in A.
  \end{equation*}
\end{lemma}

\autoref{T:regression} is valid in the general regression framework. In the proofs of \autoref{L:randomized1}, \autoref{L:randomized2}, \autoref{L:aggregated1} and \autoref{T:regression}, we consider a general regression function $\psi^\star$. Denote by $(\Theta,\mathcal{T})$ a space of functions equipped with a countable generated $\sigma$-algebra, and let $\pi$ be a probability on $(\Theta,\mathcal{T})$, referred to as the prior. \autoref{L:randomized1}, \autoref{L:randomized2}, \autoref{L:aggregated1} and \autoref{T:regression} follow from the work of \citet{B:catoni2004,A:dt2008,A:dt2012a,A:alquier,A:ab}.
Let $\delta>0$ and consider the so-called \emph{posterior} Gibbs
transition density $\rho_\delta$ with respect to $\pi$, defined as
\begin{equation}\label{eq:gibbs-posterior}
  \rho_\delta(\{\x_i,y_i\}_{i=1}^n,\psi) = \frac{\exp[-\delta r_n(\{\x_i,y_i\}_{i=1}^n,\psi)]}{\int \exp[-\delta r_n(\{\x_i,y_i\}_{i=1}^n,\psi)]\pi(\mathrm{d}\psi)}.
\end{equation}
In the following three lemmas, denote by $\rho$ a \emph{so-called} posterior probability absolutely continuous with respect to $\pi$. Let $\psi$ be a realization of a random variable $\Psi$ sampled from $\rho$.
\begin{lemma}\label{L:randomized1}
  Let \autoref{As:hyp1} and \autoref{As:hyp2} hold. Set $w=8C\max(L,C)$, $\delta \in (0,n/[w+4(\sigma^2+C^2)])$ and $\varepsilon\in(0,1)$. Then with $\Proba$-probability at least $1-\varepsilon$
  \begin{equation*}
    R(\psi) - R(\psi^\star)
    \leq \frac{1}
    {1 - \frac{4\delta(\sigma^2+C^2)}{n-w\delta}}
    \left(R_n(\psi) - R_n(\psi^\star) + \frac{\log\frac{\mathrm{d}\rho}{\mathrm{d}\pi}(\psi) + \log\frac{1}{\varepsilon}}{\delta}\right).
  \end{equation*}
\end{lemma}
\begin{proof}
  Apply \autoref{L:massart} to the variables $T_i$ defined as
  follow: for any $\psi \in $,
  \begin{equation}\label{eq:T_i}
    T_i = -(Y_i - \psi(\X_i))^2 + (Y_i - \psi^\star(\X_i))^2, \quad i\in\{1,\dots,n\}.
  \end{equation}
  First, let us note that
  \begin{align*}
    R(\psi)-R(\psi^\star) &= \E[(Y_1-\psi(\X_1))^2] - \E[(Y_1-\psi^\star(\X_1))^2] \\
    &=\E[(2Y_1-\psi(\X_1)-\psi^\star(\X_1))(\psi^\star(\X_1)-\psi(\X_1))] \\
    &=\E\left[(\psi^\star(\X_1)-\psi(\X_1))\E[(2W_1+\psi^\star(\X_1)-\psi(\X_1))|\X_1]\right] \\
    &=2\E[(\psi^\star(\X_1)-\psi(\X_1))\E[\xi_1|\X_1]]+\E[\psi^\star(\X_1)-\psi(\X_1)]^2.
  \end{align*}
  As $\E[\xi_1|\X_1]=0$,
  \begin{equation}\label{eq:pythagore}
    R(\psi)-R(\psi^\star) = \E[\psi^\star(\X)-\psi(\X)]^2.
  \end{equation}
  By \eqref{eq:T_i}, the random variables $(T_i)_{i=1}^n$ are independent. Using \autoref{L:massart}, we get
  \begin{align}
    \nonumber \sum_{i=1}^n\E T_i^2 &=
    % \sum_{i=1}^n\E\E\left[T_i^2|\X_i\right] \\
    \sum_{i=1}^n\E\left[(2Y_i-\psi(\X_i)-\psi^\star(\X_i))^2
      (\psi(\X_i)-\psi^\star(\X_i))^2\right]  \\ \nonumber
    &= \sum_{i=1}^n\E\E\left[(2W_i+\psi^\star(\X_i)-\psi(\X_i))^2
      (\psi(\X_i)-\psi^\star(\X_i))^2|\X_i\right].
      \end{align}
      Next, using that $|a+b|^k\leq 2^{k-1}(|a|+|b|)$ for any $a$, $b\in\R$ and $k\in\N^*$, we get
      \begin{align}
         \nonumber \sum_{i=1}^n\E T_i^2 &\leq 2\sum_{i=1}^n\E\left[(\psi(\X_i)-\psi^\star(\X_i))^2\E\left[(4W_i^2+4C^2)
        |\X_i\right]\right] \\ \nonumber
    &\leq 8\left(\sigma^2 + C^2\right)\sum_{i=1}^n\E\left[(\psi(\X_i)-\psi^\star(\X_i))^2\right] \\ \label{eq:def-v}
    &=8n\left(\sigma^2+C^2\right)\left(R(\psi)-R(\psi^\star)\right) \defin v,
  \end{align}
  where we have used \eqref{eq:pythagore} in the last equation. It
  follows that for any integer $k\geq 3$,
  \begin{multline*}
    \sum_{i=1}^n\E[(T_i)_+^k] = \sum_{i=1}^n\E\E[(T_i)_+^k|\X_i]
    \\
    \leq
    \sum_{i=1}^n\E\E\left[
      |2Y_i - \psi(\X_i) - \psi^\star(\X_i)|^k
      |\psi(\X_i) - \psi^\star(\X_i)|^k |\X_i\right] \\
    = \sum_{i=1}^n\E\E\left[ |2W_i + \psi^\star(\X_i) - \psi(\X_i)|^k
      |\psi(\X_i) - \psi^\star(\X_i)|^k  |\X_i\right] \\
    \leq 2^{k-1}\sum_{i=1}^n\E\E\left[\left(2^k|\xi_i|^k+|\psi^\star(\X_i) - \psi(\X_i)|^k\right)
      |\psi(\X_i) - \psi^\star(\X_i)|^k |\X_i \right].
  \end{multline*}
  Using that $|\psi(\x_i) - \psi^\star(\x_i)|^k \leq (2C)^{k-2}|\psi(\x_i) - \psi^\star(\x_i)|^2$ and \eqref{eq:pythagore}, we get
  \begin{multline*}
   \sum_{i=1}^n\E[(T_i)_+^k] \leq 2^{k-1}\sum_{i=1}^n
    \left(2^{k-1}k!\sigma^2L^{k-2}+(2C)^k\right)(2C)^{k-2}[R(\psi)-R(\psi^\star)] \\
     =
    \frac{k!}{2}v(2C)^{k-2}\left(\frac{2^{2k-4}\sigma^2L^{k-2}+\frac{2}{k!}2^{2k-4}C^k}{\sigma^2+C^2}\right).
  \end{multline*}
  Recalling that $C>\max(1,\sigma)$ gives
  \begin{align*}
   \frac{2^{2k-4}\sigma^2L^{k-2}+\frac{2}{k!}2^{2k-4}C^k}{\sigma^2+C^2} &\leq \frac{4^{k-2}\sigma^2L^{k-2}}{2\sigma^2}+\frac{\frac{2}{k!}4^{k-2}C^k}{C^2} \\
    &\leq \frac{1}{2}(4L)^{k-2}+\frac{1}{2}(4C)^{k-2} = [4\max(L,C)]^{k-2}.
  \end{align*}
  Hence
  \begin{equation}\label{eq:def-w}
  \sum_{i=1}^n\E[(T_i)_+^k]\leq \frac{k!}{2}vw^{k-2},\quad \mathrm{with}\quad  w\defin 8C\max(L,C).
  \end{equation}
  Applying \autoref{L:massart}, we obtain, for any real $\delta \in \left(0,\frac{n}{w}\right)$, with $\gamma = \frac{\delta}{n}$,
  \begin{equation*}
    \E\exp[\delta(R_n(\psi^\star) - R_n(\psi) + R(\psi) - R(\psi^\star))]
    \leq \exp\left(\frac{v\delta^2}{2n^2\left(1-\frac{w\delta}{n}\right)}\right),
  \end{equation*}
  % )}]
  that is, that for any real number $\varepsilon \in (0,1)$,
  \begin{multline}\label{eq:res}
    \E\exp\left[
      \delta[R_n(\psi^\star) - R_n(\psi)]+
      \delta[R(\psi) - R(\psi^\star)]\left(1 -
        \frac{4\delta(\sigma^2+C^2)}{n-w\delta} \right) \right. \\ \left.
      - \log\frac{1}{\varepsilon}\right]
    \leq
    \varepsilon.
  \end{multline}
  Next, we use a standard PAC-Bayesian approach (as developed in
  \citet{A:audibert2004,B:catoni2004,B:catoni2007,A:alquier}). For any prior probability $\pi$ on $(\Theta,\mathcal{T})$,
  \begin{multline*}
    \int\E\exp\left[
      \delta[R(\psi) - R(\psi^\star)]\left(1 -
        \frac{4\delta(\sigma^2+C^2)}{n-w\delta} \right)\right. \\ \left. + \delta[R_n(\psi^\star) - R_n(\psi)] - \log\frac{1}{\varepsilon}\right]\pi(\mathrm{d}\psi)
    \leq
    \varepsilon.
  \end{multline*}
  By the Fubini-Tonelli theorem
  \begin{multline*}
    \E\int\exp\left[
      \delta[R(\psi) - R(\psi^\star)]\left(1 -
        \frac{4\delta(\sigma^2+C^2)}{n-w\delta} \right)\right. \\ \left.+ \delta[R_n(\psi^\star) - R_n(\psi)] - \log\frac{1}{\varepsilon}\right]\pi(\mathrm{d}\psi)
    \leq
    \varepsilon.
  \end{multline*}
  Therefore, for any data-dependent posterior probability measure
  $\rho$ absolutely continuous with respect to $\pi$, adopting
  the convention $\infty\times 0 = 0$,
  \begin{multline}\label{eq:point-divergence}
    \E\int\exp\left[
      \delta[R(\psi) - R(\psi^\star)]\left(1 -
        \frac{4\delta(\sigma^2+C^2)}{n-w\delta}
      \right)\right. \\
    \left. + \delta[R_n(\psi^\star) - R_n(\psi)] - \log\frac{\mathrm{d}\rho}{\mathrm{d}\pi}(\psi)  - \log\frac{1}{\varepsilon}\right]\rho(\mathrm{d}\psi)
    \leq
    \varepsilon.
  \end{multline}
  Recalling that $\E$ stands for the expectation computed with respect to $\Proba$, the integration symbol may be omitted and we get
  \begin{multline*}
    \E\exp\left[
      \delta[R(\psi) - R(\psi^\star)]\left(1 -
        \frac{4\delta(\sigma^2+C^2)}{n-w\delta}
      \right)\right. \\
    \left. + \delta[R_n(\psi^\star) - R_n(\psi)] - \log\frac{\mathrm{d}\rho}{\mathrm{d}\pi}(\psi) - \log\frac{1}{\varepsilon}\right]
    \leq
    \varepsilon.
  \end{multline*}
  Using the elementary inequality $\exp(\delta x) \geq
  \1_{\R_+}(x)$, we get, with $\Proba$-probability at most $\varepsilon$
  \begin{multline*}
      \left(1 - \frac{4\delta(\sigma^2+C^2)}{n-w\delta} \right)[R(\psi) - R(\psi^\star)]
    \geq R_n(\psi) - R_n(\psi^\star) \\ + \frac{\log\frac{\mathrm{d}\rho}{\mathrm{d}\pi}(\psi)
    + \log\frac{1}{\varepsilon}}{\delta}.
  \end{multline*}
  Taking $\delta < {n}/[w+4(\sigma^2+C^2)]$ implies
  \begin{equation*}
    1 - \frac{4\delta(\sigma^2+C^2)}{n-w\delta} > 0,
  \end{equation*}
  and with $\Proba$-probability at least $1-\varepsilon$,
  \begin{equation*}
    R(\psi) - R(\psi^\star)
    \leq \frac{1}
    {1 - \frac{4\delta(\sigma^2+C^2)}{n-w\delta}}
    \left(R_n(\psi) - R_n(\psi^\star) + \frac{\log\frac{\mathrm{d}\rho}{\mathrm{d}\pi}(\psi) + \log\frac{1}{\varepsilon}}{\delta}\right).
  \end{equation*}
\end{proof}
\begin{lemma}\label{L:randomized2}
 Let \autoref{As:hyp1} and \autoref{As:hyp2} hold. Set $w=8C\max(L,C)$, $\delta \in (0,n/[w+4(\sigma^2+C^2)])$ and $\varepsilon\in(0,1)$. Then with $\Proba$-probability at least $1-\varepsilon$
    \begin{multline}\label{eq:res-lemma}
      \int R_n(\psi)\rho(\mathrm{d}\psi)-R_n(\psi^\star) \leq
      \left[1 +\frac{4\delta(\sigma^2+C^2)}{n-w\delta}\right]
      \left[\int R(\psi)\rho(\mathrm{d}\psi) \right. \\ \left. - \vphantom{\int} R(\psi^\star)\right]
      + \frac{\K(\rho,\pi)+\log\frac{1}{\varepsilon}}{\delta}.
    \end{multline}
\end{lemma}
\begin{proof}
  Set $\psi\in\F$ and
    $Z_i = (Y_i - \psi(\X_i))^2 - (Y_i - \psi^\star(\X_i))^2$, $i\in \{1,\dots,n\}$.
  Since $Z_i=-T_i$ where $T_i$ is defined in \eqref{eq:T_i}, using the same arguments that lead to  \eqref{eq:res},
  we get that for any $\delta\in (0,\in n/w)$ and $\varepsilon \in (0,1)$
%  \begin{multline*}
%    \E\exp\left[
%      \delta[R_n(\psi) - R_n(\psi^\star)] \right. \\ \left. -
%      \delta[R(\psi) - R(\psi^\star)]\left(1 +
%        \frac{4\delta(\sigma^2+C^2)}{n-w\delta} \right)
%      - \log\frac{1}{\varepsilon}\right]
%    \leq
%    \varepsilon.
%  \end{multline*}
%  The proof is along the same lines than \ref{L:randomized1}. Let $\pi$ denote some prior
%  probability on $(\Theta,\mathcal{T})$, and $\rho$ some data-dependent
%  posterior probability measure which is absolutely continuous with
%  respect to $\pi$. Previous inequality yields
  \begin{multline*}
    \E\int\exp\left[
      -\delta[R(\psi) - R(\psi^\star)]\left(1 +
        \frac{4\delta(\sigma^2+C^2)}{n-w\delta}
      \right)\right. \\
    \left. + \delta[R_n(\psi) - R_n(\psi^\star)] - \log\frac{\mathrm{d}\rho}{\mathrm{d}\pi}(\psi)  - \log\frac{1}{\varepsilon}\right]\rho(\mathrm{d}\psi)
    \leq
    \varepsilon.
  \end{multline*}
  Using Jensen's inequality, we get
  \begin{multline*}
    \E\exp\left[-\int \left\{
      \delta[R(\psi) - R(\psi^\star)]\left(1 +
        \frac{4\delta(\sigma^2+C^2)}{n-w\delta}
      \right)\right.\right. \\
    \left.\left.  + \delta[R_n(\psi) - R_n(\psi^\star)] - \log\frac{\mathrm{d}\rho}{\mathrm{d}\pi}(\psi)  - \log\frac{1}{\varepsilon}\right\}\rho(\mathrm{d}\psi)\right]
    \leq
    \varepsilon.
  \end{multline*}
  Since $\exp(\delta x) \geq \1_{\R_+}(x)$, we obtain with $\Proba$-probability at most $\varepsilon$
  \begin{multline*}
      \left[-\int R(\psi)\rho(\mathrm{d}\psi)+R(\psi^\star)\right]
      \left(1 +\frac{4\delta(\sigma^2+C^2)}{n-w\delta}\right)
      + \int R_n(\psi)\rho(\mathrm{d}\psi) \\ - R_n(\psi^\star)
      - \frac{\K(\rho,\pi)+\log\frac{1}{\varepsilon}}{\delta}
      \geq 0.
  \end{multline*}
  Taking $\delta<{n}/[w+4(\sigma^2+C^2)]$ yields \eqref{eq:res-lemma}.
\end{proof}
\begin{lemma}\label{L:aggregated1}
  Let \autoref{As:hyp1} and \autoref{As:hyp2} hold. Set $w=8C\max(L,C)$, $\delta \in (0,n/[w+4(\sigma^2+C^2)])$ and $\varepsilon\in(0,1)$. Then with $\Proba$-probability at least
  $1-\varepsilon$
  \begin{multline*}
    \int R(\psi)\rho(\mathrm{d}\psi) - R(\psi^\star)\leq \frac{1}
    {1 - \frac{4\delta(\sigma^2+C^2)}{n-w\delta}}\left(\int R_n(\psi)\rho(\mathrm{d}\psi)
      - R_n(\psi^\star) \right. \\ \left. +
      \frac{\K(\rho,\pi)+\log\frac{1}{\varepsilon}}{\delta} \right).
  \end{multline*}
\end{lemma}
\begin{proof}
  Recall
  \eqref{eq:point-divergence}. By Jensen's
  inequality,
    \begin{multline*}
      \E\exp\left[
        \delta\left(\int R(\psi)\rho(\mathrm{d}\psi) - R(\psi^\star)\right)\left[1 -
          \frac{4\delta(\sigma^2+C^2)}{n-w\delta}
        \right]\right. \\
      \left. + \delta\left(R_n(\psi^\star) -
          \int R_n(\psi)\rho(\mathrm{d}\psi)\right) -
        \K(\rho,\pi)- \log\frac{1}{\varepsilon}\right]
      \leq
      \varepsilon.
    \end{multline*}
  Using $\exp(\delta x) \geq
  \1_{\R_+}(x)$ yields the expected result.
\end{proof}
\begin{theo}\label{T:regression}
  Let $\hat{\psi}$ and
  $\hat{\psi}^{\agg}$ be realizations of the Gibbs
  estimators defined by \eqref{eq:randomized}--\,\eqref{eq:aggregated}, respectively. Let \autoref{As:hyp1} and \autoref{As:hyp2} hold. Set $w=8C\max(L,C)$ and
$\delta = n\ell/[w+4(\sigma^2+C^2)]$,
for $\ell\in(0,1)$, and let $\varepsilon\in(0,1)$. Then with probability at least $1-2\varepsilon$,
  \begin{multline}
  \left. \begin{array}{l}
R(\hat{\psi})- R(\psi^\star)
\\ R(\hat{\psi}^{\agg})- R(\psi^\star)
\end{array} \right\}
\leq \D
    \, \underset{\rho\in\MP(\Theta,\mathcal{T})}{\inf}
    \left\{\int R(\psi)\rho(\mathrm{d}\psi)\right. \\ \left. - R(\psi^\star)
    + \frac{\K(\rho,\pi) + \log\frac{1}{\varepsilon}}{n}
    \right\},
  \end{multline}
  where $\D$ is a constant depending only upon $w$, $\sigma$, $C$ and $\ell$.
\end{theo}
\begin{proof}
Recall that the randomized Gibbs estimator
  $\hat{\Psi}$ is sampled from $\rho_\delta$. Denote by $\hat{\psi}$ a realization of the variable $\hat{\Psi}$. By \autoref{L:randomized1}, with $\Proba$-probability at least $1-\varepsilon$,
  \begin{equation*}
    R(\hat{\psi}) - R(\psi^\star)
    \leq \frac{1}
    {1 - \frac{4\delta(\sigma^2+C^2)}{n-w\delta}}
    \left(R_n(\hat{\psi}) - R_n(\psi^\star) +
      \frac{\log\frac{\mathrm{d}\rho_\delta}{\mathrm{d}\pi}
        (\hat{\psi}) + \log\frac{1}{\varepsilon}}{\delta}\right).
  \end{equation*}
  Note that
  \begin{align*}
    \log\frac{\mathrm{d}\rho_\delta}{\mathrm{d}\pi}(\hat{\psi})
    &= \log\frac{\exp[-\delta R_n(\hat{\psi})]}{\int \exp[-\delta
      R_n(\psi)]\pi(\mathrm{d}\psi)} \\ &= -\delta R_n(\hat{\psi})-\log\int \exp[-\delta
    R_n(\psi)]\pi(\mathrm{d}\psi).
  \end{align*}
  Thus, with $\Proba$-probability at
  least $1-\varepsilon$,
  \begin{multline*}
    R(\hat{\psi}) - R(\psi^\star)
    \leq \frac{1}
    {1 - \frac{4\delta(\sigma^2+C^2)}{n-w\delta}}
    \left(- R_n(\psi^\star) -
      \frac{1}{\delta}\log\int \exp[-\delta
      R_n(\psi)]\pi(\mathrm{d}\psi) \right. \\ \left. +
      \frac{1}{\delta}\log\frac{1}{\varepsilon}
    \right).
  \end{multline*}
  By \autoref{L:catoni}, with
  $\Proba$-probability at
  least $1-\varepsilon$,
  \begin{multline*}
    R(\hat{\psi}) - R(\psi^\star)
    \leq \frac{1}
    {1 - \frac{4\delta(\sigma^2+C^2)}{n-w\delta}}
    \, \underset{\rho\in\MP(\Theta,\mathcal{T})}{\inf}
    \left(\int R_n(\psi)\rho(\mathrm{d}\psi)
      - R_n(\psi^\star)\right. \\ \left. +
      \frac{\K(\rho,\pi) + \log\frac{1}{\varepsilon}}{\delta}
    \right).
  \end{multline*}
  Finally, by \autoref{L:randomized2}, with $\Proba$-probability at least $1-2\varepsilon$,
  \begin{multline*}
    R(\hat{\psi}) - R(\psi^\star)
    \leq \frac{1 +\frac{4\delta(\sigma^2+C^2)}{n-w\delta}}
    {1 - \frac{4\delta(\sigma^2+C^2)}{n-w\delta}}
    \, \underset{\rho\in\MP(\Theta,\mathcal{T})}{\inf}
    \left\{\int R(\psi)\rho(\mathrm{d}\psi) - R(\psi^\star)
    \right. \\ \left.+ \frac{2}{1 +\frac{4\delta(\sigma^2+C^2)}{n-w\delta}}\frac{\K(\rho,\pi) + \log\frac{1}{\varepsilon}}{\delta}
    \right\}.
  \end{multline*}
 Apply \autoref{L:aggregated1} with the Gibbs posterior probability defined by \eqref{eq:gibbs-posterior}. With
  $\Proba$-probability at least $1-\varepsilon$,
  \begin{multline*}
    \int R(\psi)\rho_\delta(\mathrm{d}\psi) - R(\psi^\star)\leq \frac{1}
    {1 - \frac{4\delta(\sigma^2+C^2)}{n-w\delta}}\left(\int R_n(\psi)\rho_\delta(\mathrm{d}\psi)
      - R_n(\psi^\star)\right. \\ \left. +
      \frac{\K(\rho_{\delta},\pi)+\log\frac{1}{\varepsilon}}{\delta} \right).
  \end{multline*}
  Note that
  \begin{equation*}
    \begin{aligned}
      \K(\rho_{\delta},\pi)&=\int\log\frac{\exp[-\delta R_n(\psi)]}
      {\int \exp[-\delta
        R_n(\psi)]\pi(\mathrm{d}\psi)}\rho_\delta(\mathrm{d}\psi)\\
      &=-\delta\int R_n(\psi)\rho_\delta(\mathrm{d}\psi)
      - \log\left(\int \exp[-\delta
        R_n(\psi)]\pi(\mathrm{d}\psi) \right).
    \end{aligned}
  \end{equation*}
  By \autoref{L:catoni}, with $\Proba^{\otimes
    n}$-probability at least $1-\varepsilon$
  \begin{multline*}
    \int R(\psi)\rho_\delta(\mathrm{d}\psi) - R(\psi^\star)\leq \frac{1}
    {1 - \frac{4\delta(\sigma^2+C^2)}{n-w\delta}}
    \, \underset{\rho\in\MP(\Theta,\mathcal{T})}{\inf}
    \left\{\int R_n(\psi)\rho(\mathrm{d}\psi)\right. \\ \left.
      - R_n(\psi^\star) +
      \frac{\K(\rho,\pi)+\log\frac{1}{\varepsilon}}{\delta} \right\}.
  \end{multline*}
  By \autoref{L:randomized2}, with $\Proba^{\otimes
    n}$-probability at least $1-2\varepsilon$
    \begin{multline*}
      \int R(\psi)\rho_\delta(\mathrm{d}\psi) - R(\psi^\star)\leq \frac{1 +\frac{4\delta(\sigma^2+C^2)}{n-w\delta}}
      {1 - \frac{4\delta(\sigma^2+C^2)}{n-w\delta}}
      \, \underset{\rho\in\MP(\Theta,\mathcal{T})}{\inf}
      \left\{
        \int R(\psi)\rho(\mathrm{d}\psi) \right. \\
      \left. -
          R(\psi^\star)  + \frac{2}{1 +\frac{4\delta(\sigma^2+C^2)}{n-w\delta}}\frac{\K(\rho,\pi)+\log\frac{1}{\varepsilon}}{\delta}
      \right\}.
    \end{multline*}
  As $R$ is a convex function, applying Jensen's inequality gives
  \begin{equation*}
    \int R(\psi)\rho_\delta(\mathrm{d}\psi) \geq R(\hat{\psi}^{\agg}).
  \end{equation*}
  Finally, note that
  \begin{equation*}
     \frac{1+\frac{4\delta(\sigma^2+C^2)}{n-w\delta}}{1-\frac{4\delta(\sigma^2+C^2)}{n-w\delta}}=1+\frac{8\ell(\sigma^2+C^2)}{(1-\ell)(w+4\sigma^2+4C^2)}.
    \end{equation*}
\end{proof}
\begin{proof}[Proof of \autoref{T:additivemodels}]
Let $\rho\in\MP(\Theta,\mathcal{T})$. For any $A\in\mathcal{T}$, note that $\rho(A)=\sum_{\m\in\M}\rho_{\m}(A)$ where $\rho_{\m}(\cdot)=\rho(\cdot\cap \Theta_{\m})$, the trace of $\rho$ on $\Theta_{\m}$.
%  As
%   \begin{equation*}
%    \{\rho\in\MP(\Theta,\mathcal{T})\}\subset\bigcup_{\m\in\M}\{\rho\in\MP(\Theta_{\m},\B(\Theta_{\m})\},
%   \end{equation*}
%In \ref{T:regression}, as the infimum on $\rho\in\MP(\Theta,\mathcal{T})$ is smaller than the infimum on $\m\in\M$ and on $\rho\in\MP(\Theta_{\m},\B(\Theta_{\m}))$,
By \autoref{T:regression}, with $\Proba$-probability at least $1-2\varepsilon$
\begin{multline}\label{res}
      R(\hat{\psi}) - R(\psi^\star)
      \leq \D
      \, \underset{\m\in\M}{\inf}   \
      \underset{\rho\in\MP(\Theta,\mathcal{T})}{\inf} \\
      \left\{
      \int R(\psi)\rho_{\m}(\mathrm{d}\psi) - R(\psi^\star)
      + \frac{\K(\rho_{\m},\pi) + \log\frac{1}{\varepsilon}}{n}
      \right\}.
  \end{multline}
  Note that for any $\rho\in\MP(\Theta,\mathcal{T})$ and any $\m\in\M$,
    \begin{multline*}
      \K(\rho_{\m},\pi)  = \int
      \log\left(\frac{\mathrm{d}\rho_{\m}}{\mathrm{d}\pi_{\m}}
      \right)\mathrm{d}\rho_{\m} + \int
      \log\left(\frac{\mathrm{d}\pi_{\m}}{\mathrm{d}\pi}
      \right)\mathrm{d}\rho_{\m} \\
      = \K(\rho_{\m},\pi_{\m}) + \log(1/\alpha)\sum_{j\in S(\m)} m_j+\log\binom{p}{|S(\m)|} +
      \log\left(\frac{1-\left(\frac{\alpha}{1-\alpha}\right)^{p+1}}{1-\frac{\alpha}{1-\alpha}}\right).
      \end{multline*}
      Next, using the elementary inequality $\log\binom{n}{k}\leq k\log(ne/k)$ and that $\frac{\alpha}{1-\alpha}<1$,
      \begin{multline*}
      \K(\rho_{\m},\pi)\leq \K(\rho_{\m},\pi_{\m})+ \log(1/\alpha)\sum_{j\in S(\m)}m_j  +|S(\m)|\log\left(\frac{pe}{|S(\m)|}\right)%\left(\frac{pe}{|S(\m)|}\right)
      \\ +\log\left(\frac{1-\alpha}{1-2\alpha}\right).
    \end{multline*}
    We restrict the set of all probabilities absolutely continuous with respect to $\pi_{\m}$ to uniform probabilities on the ball $\B_{\m}^1(\x,\zeta)$, with $\x\in\B_{\m}^1(0,C)$ and $0<\zeta\leq C-\|\theta\|_1$. Such a probability is denoted by $\mu_{\x,\zeta}$.
  With $\Proba$-probability at least
  $1-2\varepsilon$, it yields that
  \begin{multline*}
      R(\hat{\psi}) - R(\psi^\star)
      \leq \D
      \, \underset{\m\in\M}{\inf} \, \underset{\theta\in\B_{\m}^1(0,C)}{\inf}  \
      \underset{\mu_{\theta,\zeta}, 0<\zeta\leq C-\|\theta\|_1}{\inf}
      \left\{
      \int R(\psi_{\bar{\theta}})\mu_{\theta,\zeta}(\mathrm{d}\bar{\theta})- \right. \\ \left.  R(\psi^\star)
       + \frac{1}{n}\left[
      \K(\mu_{\theta,\zeta},\pi_{\m}) + \log\frac{1}{\varepsilon}+|S(\m)|\log\left(\frac{p}{|S(\m)|}\right)+\sum_{j\in S(\m)} m_j
    \right]
      \right\}.
  \end{multline*}
  Next, note that
    \begin{align*}
      \K(\mu_{\theta,\zeta},\pi_{\m})
%      \int
%      \log\left(\frac{\mathrm{d}u_{\theta,\zeta}}      {\mathrm{d}\pi_{\m}}\right)\mathrm{d}u_{\theta,\zeta} \\
     = \log\left(\frac{V_{\m}(C)}{V_{\m}(\zeta)}\right)
   = \log\left(\frac{C}{\zeta}\right)\sum_{j\in S(\m)}m_j.
    \end{align*}
  Note also that
    \begin{align*}
      \int R(\psi_{\bar{\theta}})
      \mu_{\theta,\zeta}(\mathrm{d}\bar{\theta}) &= \int \E \left[Y_1 -
        \psi_{\bar{\theta}}(\X_1)
      \right]^2\mu_{\theta,\zeta}(\mathrm{d}\bar{\theta}) \\
      &=\int \E \left[Y_1 -
        \psi_{\theta}(\X_1) + \psi_{\theta}(\X_1) - \psi_{\bar{\theta}}(\X_1)
      \right]^2\mu_{\theta,\zeta}(\mathrm{d}\bar{\theta}),
      \end{align*}
      and
      \begin{multline*}
      \int \E \left[Y_1 -
        \psi_{\theta}(\X_1) + \psi_{\theta}(\X_1) - \psi_{\bar{\theta}}(\X_1)
      \right]^2\mu_{\theta,\zeta}(\mathrm{d}\bar{\theta}) \\ = \int R(\psi_{\theta})\mu_{\theta,\zeta}(\mathrm{d}\bar{\theta})
        + \int \E \left[\psi_{\theta}(\X_1) - \psi_{\bar{\theta}}(\X_1)
      \right]^2\mu_{\theta,\zeta}(\mathrm{d}\bar{\theta}) \\
      + 2\int \E \{ [Y_1-\psi_{\theta}(\X_1)][\psi_{\theta}(\X_1) -
        \psi_{\bar{\theta}}(\X_1)]\}\mu_{\theta,\zeta}(\mathrm{d}\bar{\theta}).
  \end{multline*}
  Since $\bar{\theta}\in \B_{\m}^1(\theta,\zeta)$,
  \begin{multline*}
      \int \E \left[\psi_{\theta}(\X_1) - \psi_{\bar{\theta}}(\X_1)
      \right]^2\mu_{\theta,\zeta}(\mathrm{d}\bar{\theta}) \\ = \int \E
      \left[\sum_{j\in
          S(\m)}\sum_{k=1}^{m_j}(\theta_{jk}-\bar{\theta}_{jk})\varphi_k(X_{1j})
      \right]^2\mu_{\theta,\zeta}(\mathrm{d}\bar{\theta}) \\
      \leq \|\theta-\bar{\theta}\|_1^2\max_{k}\|\varphi_k\|^2_\infty
      \leq \zeta^2,
  \end{multline*}
  and by the Fubini-Tonelli theorem,
  \begin{multline*}
      2\int \E \{ [Y_1-\psi_{\theta}(\X_1)][\psi_{\theta}(\X_1) -
        \psi_{\bar{\theta}}(\X_1)]\}\mu_{\theta,\zeta}(\mathrm{d}\bar{\theta}) \\
      = 2\E \left[ [Y_1 - \psi_{\theta}(\X_1)]\int [\psi_{\theta}(\X_1) - \psi_{\bar{\theta}}(\X_1)]\mu_{\theta,\zeta}(\mathrm{d}\bar{\theta}) \right] = 0,
  \end{multline*}
  since $\int
  \psi_{\bar{\theta}}(\X_1)\mu_{\theta,\zeta}(\mathrm{d}\bar{\theta}) =
  \psi_{\theta}(\X_1)$. Consequently, as
  \begin{equation*}
    \int R(\psi_{\theta})\mu_{\theta,\zeta}(\mathrm{d}\bar{\theta}) = R(\psi_{\theta}),
  \end{equation*}
  we get
  \begin{equation*}
    \int R(\psi_{\bar{\theta}})
    \mu_{\theta,\zeta}(\mathrm{d}\bar{\theta}) \leq R(\psi_{\theta}) + \zeta^2.
  \end{equation*}
  So with $\Proba$-probability at least
  $1-2\varepsilon$,
  \begin{multline*}
      R(\hat{\psi}) - R(\psi^\star)
      \leq \D
      \, \underset{\m\in\M}{\inf} \, \underset{\theta\in\B_{\m}^1(0,C)}{\inf}  \
      \underset{\mu_{\theta,\zeta}, 0<\zeta\leq C-\|\theta\|_1}{\inf}
      \Bigg\{
      R(\psi_{\theta}) + \zeta^2  - R(\psi^\star)  \\
      + \frac{1}{n}\left[ \log(C/\zeta)\sum_{j\in S(\m)} m_j+ \log\frac{1}{\varepsilon}+|S(\m)|\log\left(\frac{p}{|S(\m)|}\right)+\sum_{j\in S(\m)} m_j
    \right]
      \Bigg\}.
  \end{multline*}
  The function $t\mapsto t^2 + \log(C/t)\sum_{j\in S(\m)}m_j/n$ is convex. Its minimum is unique and is reached for $t=[\sum_{j\in S(\m)}m_j/(2n)]^{1/2}$.
  With $\Proba$-probability at least
  $1-2\varepsilon$,\vadjust{\vfill{\eject}}
 \begin{multline*}
      R(\hat{\psi}) - R(\psi^\star)
      \leq \D
      \, \underset{\m\in\M}{\inf} \, \underset{\theta\in\B_{\m}^1(0,C)}{\inf}  \
      \left\{ R(\psi_{\theta}) - R(\psi^\star) \vphantom{\frac{1}{2}} \right. \\ \left. +|S(\m)|\frac{\log(p/|S(\m)|)}{n}+\frac{\log(n)}{n}\sum_{j\in S(\m)}m_j+\frac{\log(1/\varepsilon)}{n} \right\},
  \end{multline*}
  where $\D$ is a constant depending only on $w$, $\sigma$, $C$, $\ell$ and $\alpha$.
  As the same inequality holds for $\hat{\psi}^{\agg}$, this concludes the proof.
\end{proof}
\begin{proof}[Proof of \autoref{T:sob}]
 Recall \autoref{T:additivemodels}.
  \autoref{As:hyp3} gives
  \begin{equation*}
      R(\psi_\theta)-R(\psi^\star) = \int (\psi_\theta(\mathbf{x}) -
      \psi^\star(\mathbf{x}))^2\mathrm{d}\ProbaData(\mathbf{x})
      \leq B \int(\psi_\theta(\mathbf{x}) -
      \psi^\star(\mathbf{x}))^2\mathrm{d}\mathbf{x}.
  \end{equation*}
  For any $\m\in\M$, define
  \begin{equation*}
    \psi^\star_{\m} = \sum_{j\in S^\star}\sum_{k=1}^{m_j}\theta_{jk}^\star\varphi_k.
  \end{equation*}
  To proceed, we need to check that the projection of $\theta^\star$
  onto model $\m$ lies in $\B_{\m}^1(0,C)$, \ie,
  \begin{equation*}
    \sum_{j\in S^\star}\sum_{k=1}^{m_j}|\theta_{jk}^\star|\leq C.
  \end{equation*}
  Using the Cauchy-Schwarz inequality, we get
  \begin{align*}
    \sum_{j\in S^\star}\sum_{k=1}^{m_j}|\theta_{jk}^\star| &=
    \sum_{j\in S^\star}\sum_{k=1}^{m_j}k^{r_j}|\theta_{jk}^\star|k^{-r_j} \\
    &\leq\sum_{j\in S^\star}\left[\sqrt{\sum_{k=1}^{m_j}k^{2
          r_j}(\theta_{jk}^\star)^2}\sqrt{\sum_{k=1}^{m_j}k^{-2 r_j}}\right].
  \end{align*}
  Since for any $t\geq 1$, $\sum_{k=1}^{m_j}k^{-2t}\leq
  \sum_{k=1}^\infty k^{-2t}=\pi^2/6$, the previous inequality yields
  \begin{equation*}
    \sum_{j\in S^\star}\sum_{k=1}^{m_j}|\theta_{jk}^\star|\leq
    \frac{\pi}{\sqrt{6}}\sum_{j\in S^\star}\sqrt{d_j}\leq C.
  \end{equation*}
  Recalling \eqref{eq:pythagore} and \autoref{As:hyp3}, for a $\m\in\M$ we may now write that
  \begin{multline*}
    \inf_{\substack{\theta\in\Theta_{\m}}}
    R(\psi_\theta) - R(\psi^\star) \leq R(\psi^\star_{\m}) -
    R(\psi^\star)
    \leq B\int (\psi^\star(\mathbf{x})-
    \psi^\star_{\m}(\mathbf{x}))^2\mathrm{d}\mathbf{x} \\
    = B\int \left(\sum_{j\in S^\star}\sum_{k=1+m_j}^\infty\theta_{jk}^\star\varphi_k(\mathbf{x})
    \right)^2\mathrm{d}\mathbf{x}.
  \end{multline*}
  As
  $\{\varphi_k\}_{k=1}^\infty$ forms an orthogonal basis,
  \begin{align*}
    B\int \left(\sum_{j\in S^\star}\sum_{k=1+m_j}^\infty\theta_{jk}^\star\varphi_k(\mathbf{x})
    \right)^2\mathrm{d}\mathbf{x}&=B\sum_{j\in S^\star}\sum_{k=1+m_j}^\infty(\theta_{jk}^\star)^2 \\
    &\leq B\sum_{j\in S^\star}d_j(1+m_j)^{-2r_j},
  \end{align*}
  where the normalizing numerical factors are included in the now
  generic constant $B$.
  As a consequence, with $\Proba$-probability at least
  $1-2\varepsilon$,
  \begin{multline*}
    R(\hat{\psi}) - R(\psi^\star)
    % &\leq \D \underset{\m\in\M}{\inf}
    % \left\{ B\sum_{j\in S^\star}d_j(1+m_j)^{-2r_j} + \frac{\log(n)}{n}\sum_{j=1}^p m_j+\frac{\log(p)}{n}|S(\m)|+\frac{\log\frac{1}{\varepsilon}}{n} \right\} \\
    \leq \D \underset{\m\in\M}{\inf}
    \left\{ B\sum_{j\in S^\star} \left\{d_j(1+m_j)^{-2r_j} + \frac{m_j}{n}\log(n)\right\} \right. \\ \left. +|S^\star|\frac{\log(p/|S^\star|)}{n}+\frac{\log(1/\varepsilon)}{n} \right\},
  \end{multline*}
where $\D$ is the same constant as in \autoref{T:additivemodels}. For
any $r\geq 2$, the function $t\mapsto
d_j(1+t)^{-2r_j}+\frac{\log(n)}{n}t$ is convex and admits a minimum in
$\bigl(\frac{\log(n)}{2r_j d_j n}\bigr)^{-\frac{1}{2r_j +1}}-1$.
Accordingly, choosing $m_j\sim\bigl(\frac{\log(n)}{2r_j d_j
n}\bigr)^{-\frac{1}{2r_j +1}}-1$ yields that with $\Proba$-probability
at least
  $1-2\varepsilon$,
  \begin{equation*}
    R(\hat{\psi}) - R(\psi^\star)
    \leq \D
    \left\{ \sum_{j\in S^\star} d_j^{\frac{1}{2r_j+1}}\left(\frac{\log(n)}{2nr_j}\right)^{\frac{2r_j}{2r_j+1}}
      +|S^\star|\frac{\log\left(\frac{p}{|S^\star|}\right)}{n} +\frac{\log(1/\varepsilon)}{n} \right\},
  \end{equation*}
  where $\D$ is a constant depending only on $\alpha$,
  $w$, $\sigma$, $C$, $\ell$ and $B$, and that ends the proof.
\end{proof}
\begin{proof}[Proof of \autoref{T:coro}]
  The proof is similar to the proof of \autoref{T:additivemodels}. From \eqref{res} and for any $\rho\in\MP(\Theta,\mathcal{T})$ and any $\m\in\M$,
  \begin{multline*}
    \K(\rho_{\m},\pi)
    = \K(\rho_{\m},\pi_{\m}) + \log(1/\alpha)|S(\m)|+\log\binom{p}{|S(\m)|} \\ +
    \log\left(\frac{1-\left(\alpha\frac{1-\alpha^{K+1}}{1-\alpha}\right)^{p+1}}{1-\alpha\frac{1-\alpha^{K+1}}{1-\alpha}}\right)
    + \sum_{j\in S(\m)}\log\binom{K}{|S(\m_j)|}.
  \end{multline*}
  Using the elementary inequality $\log\binom{n}{k}\leq k\log(ne/k)$ and that $\alpha\frac{1-\alpha^{K+1}}{1-\alpha}\in (0,1)$ since $\alpha<1/2$,
  \begin{multline*}
    \K(\rho_{\m},\pi)\leq \K(\rho_{\m},\pi_{\m})+ |S(\m)|\left[\log(1/\alpha) +\log\left(\frac{pe}{|S(\m)|}\right)\right]%\left(\frac{pe}{|S(\m)|}\right)
    \\ +\sum_{j\in S(\m)}|S(\m_j)|\log\left(\frac{Ke}{|S(\m_j)|}\right) + \log\left(\frac{1-\alpha}{1-2\alpha}\right).
  \end{multline*}
  Thus with $\Proba$-probability at least
  $1-2\varepsilon$,
  \begin{multline*}
    R(\hat{\psi}) - R(\psi^\star)
    \leq \D
    \, \underset{\m\in\M}{\inf} \, \underset{\theta\in\B_{\m}^1(0,C)}{\inf}  \
    \underset{\mu_{\theta,\zeta}, 0<\zeta\leq C-\|\theta\|_1}{\inf}
    \Bigg\{
    R(\psi_{\theta}) + \zeta^2  - R(\psi^\star)  \\
    + \frac{1}{n}\Biggl[ \left[\log(C/\zeta)+\log(K)\right]\sum_{j\in S(\m)} |S(\m_j)|+ \log\frac{1}{\varepsilon}+|S(\m)|\log\left(\frac{p}{|S(\m)|}\right)
    \Biggr]
    \Bigg\}.
  \end{multline*}
  Hence with $\Proba$-probability at least
  $1-2\varepsilon$,
  \begin{multline*}
    \left. \begin{array}{l}
        R(\hat{\psi})- R(\psi^\star)
        \\ R(\hat{\psi}^{\agg})- R(\psi^\star)
      \end{array} \right\}
    \leq \D \underset{\m\in\M}{\inf}\ \underset{\theta\in\B_{\m}^1(0,C)}{\inf}
    \Bigg\{ R(\psi_{\theta}) - R(\psi^\star) \vphantom{\frac{1}{2}} %\right.
    \\ %\left.
    +|S(\m)|\frac{\log(p/|S(\m)|)}{n}+\frac{\log(nK)}{n}\sum_{j\in S(\m)}|S(\m_j)|+\frac{\log(1/\varepsilon)}{n}
    \Bigg\},
  \end{multline*}
  where $\D$ is a numerical constant depending upon $w$, $\sigma$, $C$, $\ell$ and $\alpha$.
\end{proof}

\section*{Acknowledgements} The authors are grateful to G\'erard Biau
and Éric Moulines for their constant implication, and to Christophe
Giraud and Taiji Suzuki for valuable insights and comments. They also
thank an anonymous referee and an associate editor for providing constructive and helpful remarks.

% AOS,AOAS: If there are supplements please fill:
%\begin{supplement}[id=suppA]
%  \sname{Supplement A}
%  \stitle{Title}
%  \slink[url]{http://lib.stat.cmu.edu/aoas/???/???}
%  \sdescription{Some text}
%\end{supplement}

\bibliographystyle{imsart-nameyear}
% \bibliographystyle{imsart-number}
% \bibliography{/home/benjamin/Dropbox/Pro/Recherche/biblio}
\bibliography{biblio}

\end{document}